%% file: app-main.tex
\pdfoutput=1

	\documentclass[a4paper,11pt,fleqn,american]{scrartcl}
	\def\version{arxiv}

	\def\edition{final}

\input{app-preamble}

\input{app-commands}

\title{
  A Practical and Worst-Case Efficient Algorithm for Divisor Methods of Apportionment
}

	\author{%
		Raphael Reitzig%
		\footnote{%
			Department of Computer Science, 
			University of Kaiserslautern;
			\texttt{\{reitzig,\,wild\}\,@\,cs.uni-kl.de}
		} 
	\and 
		Sebastian Wild%
		\footnotemark[1]%
	}

\ifarxiv{%
	\usepackage{hyperref}
}

\makeatletter
\AtBeginDocument{%
	{
		\def\and{and }
		\def\footnote#1{}
		\hypersetup{
		  pdftitle={\@title},
		  pdfauthor={Raphael Reitzig and Sebastian Wild}
		}
	}
}
\makeatother

\begin{document}

\maketitle

\begin{abstract}
	Proportional apportionment is the problem of assigning seats to parties
	according to their relative share of votes. Divisor methods are the de"=facto
	standard solution, used in many countries.
	
	In recent literature, there are two algorithms that implement divisor methods:
	one by \textcite{Cheng2014} has worst"=case optimal running time but is	complex, 
	while the other~\cite{Pukelsheim2014} is relatively simple and fast in practice 
	but does not offer worst"=case guarantees.
	
	We demonstrate that the former algorithm is much slower than the other in 
	practice and propose a novel algorithm that avoids the shortcomings of both. 
	We investigate the running"=time behavior of the three contenders in order to
	determine which is most useful in practice.
\end{abstract}

\input{app-introduction}
\input{app-notation}
\input{app-algorithm}
\input{app-comparison}

\section{Conclusion}
We have derived an algorithm implementing divisor methods of apportionment
that is worst"=case efficient, simple and practicable.
As such, it does not have the shortcomings of previously known algorithms. 
Even though it can not usually outperform \Pukalg, its robustness against
changing parameters makes it a viable candidate for use in practice.

\section*{Acknowledgments}
We thank Chao Xu for pointing us towards the work by \textcite{Cheng2014} and
noting that the problem of envy"=free stick"=division~\cite{ReitzigWild2015} 
is related to proportional apportionment as discussed there. 
He also observed that our approach for cutting sticks~-- the core ideas of which 
turned out to carry over to this article~-- could be improved to run in linear 
time.

Furthermore, we owe thanks to an anonymous reviewer whose constructive feedback
sparked broad changes which have greatly improved the article over its first 
incarnation.

	\printbibliography

\clearpage
\appendix
\input{app-methods-scope}
\input{app-proofs}
\input{app-algorithm-reviews}
\input{app-machine-spec}
\input{app-moreplots}

\ifarxiv{\clearpage}
\input{app-notation-reference}
\input{app-errata}

\end{document}

%% file: app-preamble.tex
\usepackage{ifthen}
\def\arxivversion{arxiv}
\def\localversion{local}
\def\lncsversion{lncs}
\newcommand\ifarxiv[2][]{\ifthenelse{\equal{\version}{\arxivversion}}{#2}{#1}}
\newcommand\iflocal[2][]{\ifthenelse{\equal{\version}{\localversion}}{#2}{#1}}
\newcommand\iflncs[2][]{\ifthenelse{\equal{\version}{\lncsversion}}{#2}{#1}}

\def\draftedition{draft}
\def\finaledition{final}
\newcommand\ifdraft[2][]{\ifthenelse{\equal{\edition}{\draftedition}}{#2}{#1}}
\newcommand\iffinal[2][]{\ifthenelse{\equal{\edition}{\finaledition}}{#2}{#1}}

\usepackage[utf8]{inputenc}
\usepackage[T1]{fontenc}
\usepackage{lmodern}

\usepackage{marginnote,changepage,relsize,xspace,nameref,enumitem,adjustbox}
\usepackage{booktabs,multirow,siunitx}

  \usepackage{microtype}

	\addtokomafont{caption}{\sffamily\small}
	\addtokomafont{captionlabel}{\sffamily\textbf}
	\setcapmargin{2em}

\usepackage{babel}
\input{ushyphex.tex} 

	\global\catcode`"=13
	\global\def"#1{%
		\ifthenelse{\equal{#1}{=}}{-}{}%
		\ifthenelse{\equal{#1}{/}}{-}{}%
		\ifthenelse{\equal{#1}{~}}{-}{}%
	}

\usepackage{amsmath,amssymb,nicefrac,mathtools,bm,colonequals}
\allowdisplaybreaks[3]

\usepackage{tikz}
\usetikzlibrary{positioning,arrows,decorations.pathreplacing,calc,shapes}
\ifarxiv[ 
  \usepackage{pgfplots,pgfplotstable}
  \usepgfplotslibrary{statistics}
  \pgfplotsset{compat=1.12}
]{
  \newcommand{\pgfplotstableread}[2]{}
}

\pgfdeclarelayer{background}
\pgfdeclarelayer{foreground}
\pgfsetlayers{background,main,foreground}

\usetikzlibrary{external}
\tikzset{
	external/export=true,
	external/mode=list and make
}
\newcommand{\plotref}[1]{%
  \ref{#1}%
  \tikzexternalenable
}
\providecommand*{\tikzname}[1]{%
  \ifdraft{\marginnote{%
	  \rlap{\textsmaller{\texttt{\detokenize{#1}}}}%
  }}{}
  \tikzsetnextfilename{#1}%
}

	\usepackage{csquotes}
	\usepackage[backend=bibtex,bibstyle=alphabetic,citestyle=alphabetic]{biblatex}
  \renewbibmacro{url}{%
    \iffieldundef{eprint}{\iffieldundef{doi}{%
      \iffieldundef{url}{}{%
        \printfield{url}
      }
    }{}}{}
  }

\usepackage{wref}
\ifarxiv[{
  \usepackage[
	  final,
	  unicode=true, 
	  bookmarks=true,
	  bookmarksnumbered=true,
	  bookmarksopen=true,
	  breaklinks=true,
	  pdfborder={0 0 0}, 
	  colorlinks=false   
  ]{hyperref}
}]{
}

\makeatletter
	\usepackage[amsmath,hyperref,thmmarks]{ntheorem}
	\theorembodyfont{\slshape}
	\theoremseparator{:}

	\newtheoremstyle{proofstyle}%
	  {\item[\theorem@headerfont\hskip\labelsep ##1\theorem@separator]}%
	  {\item[\theorem@headerfont\hskip\labelsep ##1 of ##3\theorem@separator]}

	\qedsymbol{\raisebox{-.25ex}{$\Box$}}

\makeatother

%% file: app-introduction.tex
\section{Introduction}

The problem of proportional apportionment arises whenever we have a finite
supply of $k$ indivisible, identical resource units which we have to 
distribute across $n$ parties \emph{fairly}, that is according to the proportional
share of publicly known and agreed"=upon values $v_1, \ldots, v_n$ (of the sum
$V = \sum v_i$ of these values). 
We elaborate in this section on applications of and solutions for this problem.

\ifarxiv[\medskip\noindent]{\clearpage}%
Apportionment arises naturally in politics. 
Here are two prominent examples:
\begin{itemize}
\item
	In a proportional"=representation electoral system 
	we have to assign seats in parliament to political parties 
	according to their share of all votes.
	
	The resources are seats, and the values are vote counts.
\item
	In federal states the number of representatives from each component state 
	often reflects the population of that state, even though
	there will typically be at least one representative for any state no matter 
	how small it is.
	
	Resources are again seats, values are the numbers of residents.
\end{itemize}
In order to use consistent language throughout this article, 
we will stick to the first metaphor.
That is, we assign $k$~\emph{seats} to \emph{parties}~$[1..n]$ 
proportionally to their respective \emph{votes}~$v_i$, 
and we call~$k$ the \emph{house size}.

A fair allocation should assign $v_i/V$ seats to party $i$,
where $V = v_1 + \cdots + v_n$ is the total vote count of all parties.
In case of electoral systems which exclude parties below a certain threshold of 
overall votes from seat allocation altogether, 
we assume they have already been removed from our list of $n$ parties.

As seats are indivisible, this is only possible if, by chance, 
all $v_i/V$ are integers; 
otherwise we have to come up with some rounding scheme.
This is where \emph{apportionment methods} come into play.
The books by \textcite{BalinskiYoung2001} and \textcite{Pukelsheim2014}
give comprehensive introductions into the topic with its historical,
political and mathematical dimensions.

Mathematically speaking, an apportionment method is a function 
$\appfunc : \R_{>0}^n \times \N \to \N_0^n$ that maps
vote counts $\vect v = (v_1,\ldots,v_n)$ and house size~$k$ to
a seat allocation $\vect s = (s_1,\ldots,s_n) \ce \appfunc(\vect v, k)$
so that $s_1 + \cdots + s_n = k$.
We interpret $\vect s$ as party $i$ getting $s_i$ seats.

There are many conceivable such methods, 
but there are at least three natural properties one would like 
apportionment systems to have:
\begin{enumerate}[label=(P\arabic*),leftmargin=3em,itemsep=1ex]
\item
	\label{prop:vote-monotonicity}
	\textsl{Pairwise vote monotonicity:} 
	When votes change, \appfunc should not take away seats from a party 
	that has gained votes while at the same time awarding seats to one 
	that has lost votes.
\item
	\label{prop:house-monotonicity}
	\textsl{House monotonicity:}
	\appfunc should not take seats away from any party when
	the house grows (in number of seats) but votes do not change.
\item
	\label{prop:quota-rule}
	\textsl{Quota rule:}
	The number of seats of each party should be its proportional share, rounded either up or down.
\end{enumerate}
\Citeauthor{BalinskiYoung2001} have shown that
\begin{itemize}
  \item \ref{prop:vote-monotonicity} implies \ref{prop:house-monotonicity}
    \cite[Cor.\ 4.3.1]{BalinskiYoung2001},
  
  \item no method can always guarantee \ref{prop:vote-monotonicity} \emph{and}
    \ref{prop:quota-rule} \cite[Thm.\ 6.1]{BalinskiYoung2001}, and
  
  \item \ref{prop:vote-monotonicity} holds exactly for \emph{divisor methods} \cite[Thm.\ 4.3]{BalinskiYoung2001}.
\end{itemize}
Property~\ref{prop:vote-monotonicity} is essential for upholding the principle of
``one-person, one-vote'', an ideal pursued by electoral systems around the globe 
and occasionally enforced by law~\cite[Section~2.4]{Pukelsheim2014}.
Therefore, divisor a.\,k.\,a.\ Huntington methods can be the only choice, 
for the price of \ref{prop:quota-rule}.
Other choices can be made, of course; the aforementioned books~\cite{BalinskiYoung2001,Pukelsheim2014} discuss different trade"=offs.

Divisor methods are characterized by \emph{divisor sequences} which control
the notion of ``fairness'' implemented by the respective method. 
There are many popular choices (cf.\ \wref{tab:divisor-sequences}).
It is not per se clear which divisor sequence is the best; there still seems
to be active discussion, e.\,g., for the U.\,S.\ House of Representatives.
One reason is that no-one has yet been able to propose a convincing, universally 
agreed"=upon mathematical criterion that would single out one method as superior 
to the others. In fact, there are competing notions of fairness, 
each favoring a different divisor method~\cite[Section~A.3]{BalinskiYoung2001}.
A reasonable approach is therefore to run computer simulations of different methods
and compare their outcomes empirically, for example w.\,r.\,t.\ the distribution
of final average votes per seat $v_i/s_i$.
For this purpose, many apportionments have to be computed, 
so efficient algorithms can become an issue.

\begin{table}[tb]
	\begin{center}%
		\def\nobounds{---}%
		\def\rowsep{1ex}
		\def\mslcases{\tcases{2x+1}{x\ge1}{1.6x+1.4}{x<1}}%
		\begin{adjustbox}{max width=\linewidth}
		  \begin{tabular}{p{4cm}lcc}
			  \toprule
			    \bfseries{Method}        
			  & \bfseries{Divisor Sequence} 
			  & \bfseries{\boldmath $\delta(x)$}    
			  & \bfseries Sandwich               \\ 
			  \midrule%
			
			    Smallest divisors
			  & 0, 1, 2, 3, \ldots
			  & $x$ 
			  & \nobounds \\[\rowsep]%
			
			    Greatest divisors
			  & 1, 2, 3, 4, \ldots 
			  & $x+1$
			  & \nobounds \\[\rowsep]%
			
			    Sainte-Laguë
			  & 1, 3, 5, 7, \ldots
			  & $2x+1$ 
			  & \nobounds \\[\rowsep]%
			
			    Modified Sainte-Laguë
			  & 1.4, 3, 5, 7, \ldots  
        & \;\;$\mslcases$\;\;
        & $2x+\frac65 \pm\frac15$ \\[\rowsep]%
        
			    Equal Proportions
			  & 0, $\sqrt2$, $\sqrt6$, $\sqrt{12}$, \ldots
			  & $\sqrt{x(x+1)}$
			  & $\phantom{2}x+\frac14\pm\frac14$ \\[\rowsep]%
			
			    Harmonic Mean
			  & 0, $\frac43$, $\frac{12}5$, $\frac{24}7$, \ldots
			  & $\frac{2x(x+1)}{2x+1}$    
			  & $\phantom{2}x+\frac14\pm\frac14$ \\[\rowsep]%
			
			    Imperiali
			  & 2, 3, 4, 5, \ldots
			  & $x+2$
			  & \nobounds \\[\rowsep]%
			
			    Danish
			  & 1, 4, 7, 10, \ldots
			  & $3x+1$
			  & \nobounds \\%
			
			  \bottomrule
		  \end{tabular}%
		\end{adjustbox}
	\end{center}
	\caption{%
		Commonly used divisor methods \cite[Table~1]{Cheng2014}.
		For each of the methods, we give a possible continuation 
		$\delta$ of the respective divisor sequence (cf.\ \wref{sec:notation}) 
		as well as linear sandwich bounds on $\delta$, if non-trivial
		(cf.\ \wref{lem:candidate-set-linear-corridor}).
	}
	\label{tab:divisor-sequences}
\end{table}

We thus study the problem of computing the final seat allocation by divisor 
methods (given by their divisor sequences) according to vote counts 
and house size. 

For the case of almost linear divisor sequences, the problem can be solved
in time $\Oh(n)$; this has been shown by \textcite{Cheng2014} who propose
a worst"=case running"=time"=optimal algorithm which we call \CEalg. 
It is quite involved and rather difficult to implement 
(cf.\ \wref{app:CE-main-procedure}).

\Textcite{Pukelsheim2014}, on the other hand, proposes algorithm \Pukalg whose 
running time is not asymptotically optimal in the worst case but tends
to perform well in practice, at least if some insight about the used 
divisor sequence is available and inputs are good"=natured 
(cf.\ \wref{app:jump-and-step-review}).

After introducing divisor methods formally in \wref{sec:notation},
we propose a new algorithm in \wref{sec:fast-selection-alg} that also attains
the $\Oh(n)$ worst"=case running time bound but is straight"=forward to implement
\emph{and} efficient in practice as well.
It is based on a generalization of our solution for the envy"=free 
stick"=division problem~\cite{ReitzigWild2015}.

We finally compare the performance of the three contending algorithms with extensive
running time experiments, an executive summary of which we give in
\wref{sec:algorithms-comparison}.

Additional material includes an index of notation in \wref{app:notations}.

%% file: app-notation.tex
\section{Divisor Methods Formalized}
\label{sec:notation}

Let $d = (d_j)_{j=0}^\infty$ be an arbitrary divisor sequence, i.\,e.\ a
nonnegative, strictly increasing and unbounded sequence of real numbers.
We formally set $d_{-1} \ce -\infty$.

We require that there is a smooth continuation of $d$ on the reals 
which is easy to invert. 
That is, we assume a function $\delta : \R_{\ge 0} \to \R_{\ge d_0}$ with
\begin{enumerate}[label=\roman*)]
	\item\label{item:delta-continc}%
	  $\delta$ is continuous and strictly increasing,
	\item\label{item:deltainv-constant}%
	  $\delta^{-1}(x)$ for $x\ge d_0$ can be computed 
		with a constant number of arithmetic operations, and
	\item\label{item:delta-is-d}%
	  $\delta(j) = d_j$ (and thus $\delta^{-1}(d_j) = j$) for all $j \in \N_0$.
\end{enumerate}
All the divisor sequences used in practice fulfill these requirements;
cf.\ \wref{tab:divisor-sequences}.
For convenience, we continue $\delta^{-1}$ on the complete real line requiring
\begin{enumerate}[label=\roman*),start=4]
  \item\label{item:delta-inverse-below-d0}%
	  $\delta^{-1}(x) \in [-1,0)$ for $x < d_0$.
\end{enumerate}

\begin{corollary}\label{cor:delta-inverse-monotonicity}
  Assuming \ref{item:delta-continc} to \ref{item:delta-inverse-below-d0},
	$\delta^{-1}(x)$ is continuous and strictly increasing on $\R_{\ge d_0}$.
	Furthermore, it is the inverse of $j \mapsto d_j$ in the sense that
	  \[   \lfloor \delta^{-1}(x) \rfloor
	     = \max \{ j \in \Z_{\geq -1} \mid d_j \leq x \} \]
	for all $x \in \R$.
\qed\end{corollary}
In particular, $\lfloor \delta^{-1}(x) \rfloor = j$ for $d_j \leq x < d_{j+1}$
so the floored $\delta^{-1}$ is the (zero"=based) \emph{rank} function for the
set of all $d_j$ as long as $x \geq d_0$.

Note how this reproduces what is called \emph{$d$"/rounding} in the
literature~\cite{BalinskiYoung2001,Pukelsheim2014}; we obtain an efficient way 
of calculating this function via $\delta^{-1}$.

Now the set of all seat assignments that are valid w.\,r.\,t.\ $d$ 
is given by~\cite{BalinskiYoung2001}
  \[ \mathcal{S}(\vect v, k) 
       \wwrel= \Bigl\{ \vect s \in \N_0^n \ \Bigm|\  
              \sum_{i=1}^n s_i = k 
                \wrel\land 
              \exists\,a > 0.\ \forall\,i \in [1..n].\ 
                s_i \in \lfloor \delta^{-1}(v_i \cdot a) \rfloor + \{0,1\} 
         \Bigr\}. \]
We call a realization of~$a$ \emph{proportionality constant}~$a^*$;
intuitively, every seat corresponds to roughly $\nicefrac{1}{a^*}$ votes.

An equivalent definition is by the set of possible results of the following
algorithm~\cite[Prop.\ 3.3]{BalinskiYoung2001}.
\begin{algorithm}\label{alg:highest-averages}
  $\textsc{\highestavgsalgASCII}_d(\vect v, k)$ :
  \begin{enumerate}[label={\textsf{\textbf{Step~\arabic*}}},ref=\arabic*,leftmargin=4.5em,itemsep=0ex]
    \item\label{step:highest-averages-init}%
      Initialize $\vect s = 0^n$.
      
    \item\label{step:highest-averages-while}%
      While $k > 0$,
      \begin{enumerate}[label={\textsf{\textbf{Step~\arabic{enumi}.\arabic*}}},ref=\arabic{enumi}.\arabic*,leftmargin=2em,topsep=0ex,itemsep=0ex]
        \item\label{step:highest-averages-max}%
          Determine $I = \arg\min_{i=1}^n d_{s_i} / v_i$.
        \item\label{step:highest-averages-update}%
          Update $s_{I} \leftarrow s_{I} + 1$ and $k \leftarrow k-1$.
      \end{enumerate}
     
    \item\label{step:highest-averages-return}%
      Return $\vect s$.
  \end{enumerate}
\end{algorithm}
We can obtain a proportionality constant~\cite[59f]{Pukelsheim2014} by
\begin{equation}\label{eq:seats-to-astar}%
  a^* = \max \{ d_{s_i - 1} / v_i \mid 1 \leq i \leq n \},
\end{equation}
which in turn defines the set $\mathcal{S}(\vect v, k)$.

Note that we work with $d_j/v_i$ instead of $v_i/d_j$ in the classical literature;
\textcite{Cheng2014} and we prefer the reciprocals because the case $d_0 = 0$ 
then handles gracefully and without special treatment. 
Therefore, our $a^*$ is also the reciprocal of the proportionality constant as
e.\,g.\ \textcite{Pukelsheim2014} defines it, we multiply by $a$ in 
the definition of $\mathcal{S}$ and we take the minimum in \HAalg.
It is important to note that the defined set $\mathcal{S}$ remains unchanged
by this switch.

Following the notation of \textcite{Cheng2014}, we furthermore define
for given votes $\vect v = (v_1,\ldots,v_n) \in \Q_{>0}^n$ the sets
  \[ A_i 
	     \wrel\ce \Bigl\{ a_{i,j} \Bigm| j=0,1,2,\ldots  \Bigr\}
      	\wwrel{\text{with}} a_{i,j} \ce \frac{d_j}{v_i} \]
and their multiset union
  \[ \mset A \wrel\ce \biguplus_{i=1}^n A_i . \]
As we will see later, the relative rank of elements in $\mset A$ turns out to
be of interest; we therefore define the \emph{rank function} $r(x,\mset A)$ which denotes 
the number of elements in multiset $\mset A$ that are no larger than~$x$, that is
\begin{equation}\label{eq:rank-as-sum}
  r(x,\mset A)
	  \wwrel\ce
	\bigl|\mset A \cap (-\infty,x] \, \bigr|
	  \wwrel=
  \sum_{i=1}^n |\{a_{i,j} \in \mset{A} \mid a_{i,j} \le x\}| . 
\end{equation}
We write $r(x)$ instead of $r(x,\mset A)$ when $\mset A$ is clear from context.

We need two more convenient shorthands: Assuming we have $a^* \leq \overline{x}$,
we denote with
\begin{equation}\label{eq:drop-irrelevant-parties}
  I_{\overline x} \wwrel\ce 
    \bigl\{ i \in \{1,\dots,n\} \mid v_i > d_0/\overline x \bigr \}
\end{equation}
the set of parties that can hope for a seat, and with
\begin{equation}\label{eq:one-sided-candidate-cutoff}
   \mset{A}^{\overline x} 
     \wwrel\ce 
       \biguplus_{i \in I_{\overline x}}
			   \biggl\{ \frac{d_j}{v_i} \in \mset{A}  
			     \biggm| \frac{d_j}{v_i} < \overline x 
			   \biggr\} 
		 \wwrel=
       \biguplus_{i = 1}^n
			   \biggl\{ \frac{d_j}{v_i} \in \mset{A}  
			     \biggm| \frac{d_j}{v_i} < \overline x 
			   \biggr\} 
		 \wwrel=
		   \mset A \cap (-\infty,\overline x) 
\end{equation}
the multiset of elements from sequences of these parties that are smaller 
than $\overline x$, i.\,e.\ reasonable candidates for $a^*$.

%% file: app-algorithm.tex
\section{Fast Apportionment by Rank Selection}%
\label{sec:fast-selection-alg}

From \wref{eq:seats-to-astar} together with strict monotonicity of $d$,
we obtain immediately that $a^* = \mset{A}_{(k)}$, i.\,e.\ the $k$th~smallest 
element of $\mset{A}$ (counting duplicates) is a suitable proportionality
constant. This allows us to switch gears from the iteration"=based world
of \textcite{Pukelsheim2014} to selection"=based algorithms, as previously
seen by \textcite{Cheng2014}.

Note that even though $\mset A$ is infinite, $\mset A_{(k)}$ always exists
because the terms $a_{i,j} = \nicefrac{d_j}{v_i}$ are strictly increasing in~$j$ 
for all $i \in \{1,\ldots,n\}$.

Borrowing terminology from the field of mathematical optimization,
we call $a$ \emph{feasible} if $r(a) \ge k$, otherwise it is \emph{infeasible}.
Feasible $a \ne a^*$ are called \emph{suboptimal}. Our goal is to find a subset
of $\mset{A}$ that contains $a^*$ but as few infeasible or suboptimal $a$ as
possible; we can then apply a rank"=selection algorithm on this subset and
obtain (via $a^*$) the solution to the apportionment problem.

Now since $d$ is unbounded, setting \emph{any} upper bound $\overline{x}$ on 
the $a_{i,j}$ yields a finite search space $\mset{A}^{\overline{x}}$. 
By choosing any such bound that maintains $|\mset{A}^{\overline{x}}| \geq k$, 
we retain the property that $a^*$ is the $k$th~smallest element under consideration.

One naive way is to make sure that the party with the most votes (which should
get the most seats) contributes at least $k$~values to $\mset A$. This can be
achieved by letting $\overline{x} = d_{k-1} / \max \vect v + \varepsilon$
(cf.\ the proof of \wref{thm:correctness-and-time-for-linear-alg}).
This alone, however, leads only to an algorithm with worst"=case running time 
in $\Th(kn)$, which is worse than even \HAalg (with priority queues).

We can actually not improve this upper bound $\overline x$; it is tight for the
case that one party has many more votes than all others and gets (almost) all
of the seats.
We can, however, exclude many individual elements in $\mset{A}^{\overline x}$ 
because they are too small to be feasible or too large to be optimal.

Towards finding suitable upper and lower bounds on $a^*$, we investigate its
rank in the multiset $\mset{A}$ of all candidates. All we know is that
  \[ k \wwrel\leq r(a^*) \wwrel\leq k + |I_{\overline x}| \]
since we may have any number between one and $|I_{\overline x}|$ parties that
tie for the last seat. 
We can still make an ansatz with 
  $r(\overline a) \geq k + |I_{\overline x}|$ 
and
  $r(\underline a) < k$,
express rank function~$r$ in terms of $\delta^{-1}$ 
(cf.\ \wref{lem:rank-as-sum} in \wref{app:lemmata}) and derive that
\begin{equation}\label{eq:suff-cond-for-sandwich}
    \sum_{i \in I_{\overline x}} 
	    \delta^{-1}(v_i \cdot \underline a) 
	  \wwrel\le k - |I_{\overline x}|
  \qquad\text{and}\qquad
	  \sum_{i \in I_{\overline x}} \delta^{-1}(v_i\cdot\overline a)
	  \wwrel\ge k.
\end{equation}
This pair of inequalities is indeed a sufficient condition for admissible 
pairs of bounds $(\underline a, \overline a)$; we can conclude that 
$\underline a \leq a^* \leq \overline a$. For a formal proof, 
see \wref{lem:astar-corridor} in \wref{app:lemmata}.

We now want to derive a sandwich on $a^*$ by fulfilling the inequalities
in \wref{eq:suff-cond-for-sandwich} as tightly as possible. 
Depending on $\delta^{-1}$, this may be hard to do analytically.
However, we can make the same assumption as \textcite{Cheng2014} and explicitly 
compute suitable bounds for divisor sequences which behave roughly linearly.
This does not limit the scope of our investigation by much;
see \wref{app:methods-scope} for more on this.

\begin{lemma}
\label{lem:candidate-set-linear-corridor}
	Assume the continuation $\delta$ of divisor sequence $d$ fulfills
	\[
		         \alpha x + \underline\beta
		\wrel\le \delta(x) 
		\wrel\le \alpha x + \overline\beta
	\]
	for all $x \in \R_{\ge 0}$ with $\alpha > 0$, $\underline\beta \in [0, \alpha]$
	and $\overline\beta \geq 0$.
	Let further some $\overline x > a^*$ be given. 
	Then, the pair $(\underline a, \overline a)$ defined by
	\[
			\underline a  
		\wrel\ce
			\max\biggl\{0,
				\frac{\alpha k - (\alpha - \underline\beta) \cdot |I_{\overline x}|}%
				     {V_{\overline x}}
			\biggr\}
		\wwrel{\text{and}}
			\overline a
		\wrel\ce 
			\frac{\alpha k + \overline\beta \cdot |I_{\overline x}|}%
			     {V_{\overline x}}
	\]
	with $V_{\overline x} \ce \sum_{i \in I_{\overline x}} v_i$
	fulfills the conditions of \wref{lem:astar-corridor}, that is
	$\underline a \le a^* \le \overline a$.
	Moreover, 
	\[
			\bigl|\mset A \cap [\underline a,\overline a] \big|
		\wrel\le 
			      2\biggl(1 + \frac{\overline\beta - \underline\beta}\alpha \biggr)
			\cdot |I_{\overline x}|.
	\]
\end{lemma}
The proof consists mostly of rote calculation towards applying \wref{lem:astar-corridor};
see \wref{app:lemmata} for the details.

We have now derived our main improvement over the work by \textcite{Cheng2014}; 
where they have only a one"=sided bound on $a^*$ and thus have to employ an 
involved search on $\mset A$, we have sandwiched $a^*$ from both sides, and so 
tightly that the remaining search space is small enough for a simple rank selection 
to be efficient.

Building on the bounds from \wref{lem:candidate-set-linear-corridor},
we can improve upon the naive idea using only $\overline x$ 
by excluding also many more elements from $\mset A$ which are for sure not $a^*$. 
Since we remove in particular too small elements, this means that we also have 
to modify the rank we select; we will see that our bounds are chosen so that we 
can use $\delta^{-1}$ to \emph{count} the number of elements we discard \emph{exactly}.

Recall that we assume a fixed apportionment scheme, that is fixed $d$ with
known $\alpha$, $\underline\beta$ and $\overline\beta$ as per
\wref{lem:candidate-set-linear-corridor}.
\begin{algorithm}
  \label{alg:selection-algo-linear-candidates}
  $\textsc{\RWalgASCII}(\vect v, k)_d$ :
  \begin{enumerate}[label={\textsf{\textbf{Step~\arabic*}}},ref=\arabic*,leftmargin=4.5em,itemsep=0ex]
    \item \label{step:linear-selection-find-v1}
	    Find the $v^{(1)} = \max \{ v_1,\ldots,v_n \}$.
	    
    \item \label{step:linear-selection-set-overline-x}
	    Set $\overline x \ce d_{k-1} / v^{(1)} + \varepsilon$
	    for suitable\footnote{%
	      Neither correctness nor $\Th$"/running"=time is affected by the choice
	      of $\varepsilon$ here since it affects only the size of $I_{\overline x}$,
	      which is bounded by $n$ in any case. In particular, the size of 
	      $\hat{\mset A}$ is affected only up to a constant factor. 
	      For tweaking performance in practice,
	      see the proof of \wref{thm:correctness-and-time-for-linear-alg}.} 
	    constant $\varepsilon > 0$.
    
    \item \label{step:linear-selection-compute-I}
	    Compute 
	    $I_{\overline x}$ 
	    as per \wref{eq:drop-irrelevant-parties}.
	    
    \item \label{step:linear-selection-compute-a-bounds}
	    Compute 
			    $\underline a$
	    and 
			    $\overline a$
	    as per \wref{lem:candidate-set-linear-corridor}.
	    
	  \item \label{step:linear-selection-resultinit}
	    Initialize $\hat{\mset A} \ce \emptyset$ and $\hat k \ce k$.
    \item \label{step:linear-selection-compute-A-hat-and-k-hat-loop}
	    For all $i \in I_{\overline x}$, do:%
	    
	    \begin{enumerate}[label={\textsf{\textbf{Step~\arabic{enumi}.\arabic*}}},ref=\arabic{enumi}.\arabic*,leftmargin=2em,topsep=0ex,itemsep=0ex]
	    \item \label{step:linear-selection-compute-jmin-jmax}
		    Compute $\underline j \ce \max\bigl\{0,
									    \bigl\lceil \delta^{-1}( v_i\cdot \underline a ) \bigr\rceil
								    \bigr\}$ and
		    $\overline j \ce \lfloor \delta^{-1}(v_i \cdot \overline a) \rfloor $.
		    
	    \item \label{step:linear-selection-enlarge-A-hat}
		    Add all $d_j / v_i$ to $\hat{\mset A}$
		    for which $\underline j \le j \le \overline j$.
		    
	    \item \label{step:linear-selection-adapt-k-hat}
		    Update $\hat k \leftarrow \hat k - \underline j$. 
	    \end{enumerate}

    \item \label{step:linear-selection-select-A-hat}
	    Select and return $\hat{\mset A}_{(\hat k)}$.
  \end{enumerate}
\end{algorithm}

\begin{theorem}
\label{thm:correctness-and-time-for-linear-alg}
	\wref{alg:selection-algo-linear-candidates} computes $a^*$ in time $\Oh(n)$ 
	for any divisor sequence~$d$ that fulfills the requirements of
	\wref{lem:candidate-set-linear-corridor}.
\end{theorem}
\begin{proof}
First, we have to show that $I_{\overline x}$ as we compute it in
Steps~\ref{step:linear-selection-find-v1}-\ref{step:linear-selection-compute-I} 
is correct. 
We have $\overline x > a^* = \mset A^{\overline x}$ as already 
$r(\overline x - \varepsilon) = r(d_{k-1} / v^{(1)}) \ge k$; 
at least the $k$ elements 
$\frac{d_0}{v^{(1)}},\ldots,\frac{d_{k-1}}{v^{(1)}} \in \mset A$
are no larger than $d_{k-1} / v^{(1)}$. 
We thus never need to consider elements $a\ge \overline x$, and
in particular $\mset A_{(k)} = \smash{\mset A^{\overline x}_{(k)}}$ as
$\mset A^{\overline x} = \mset A\cap (-\infty,\overline x)$.

So far, we have needed no additional restriction on $\varepsilon$ in 
\wref{step:linear-selection-set-overline-x}; we only need it to be positive
so we do not discard $a^*$ by accident if it is exactly $d_{k-1}/v^{(1)}$.
However, the size of $\mset{A}^{\overline x}$ can be arbitrarily large~--
depending on the input values~$v_i$ which we do not want. Therefore, we
require
\begin{equation}\label{eq:suitable-epsilon}
  0 < \varepsilon < \frac{d_k - d_{k-1}}{v^{(1)}};
\end{equation}
such exists because $d$ is strictly increasing.
Note how then $\overline{x} < d_k/v^{(1)}$ so we do not keep any
additional suboptimal values.

From \wref{step:linear-selection-compute-a-bounds} on, we then construct
multiset $\hat{\mset A} \subseteq \mset A$ as the subsequent union of
$A_i \cap [\underline a, \overline a]$, that is
\begin{align*}
		\hat{\mset A}
	&\wwrel=
	  \biguplus_{i \in I_{\overline x}}
	  \biggl\{ 
			\frac{d_j}{v_i} \biggm|  \underline j(i) \le j \le \overline j(i)
		\biggr\} \\
  &\wwrel=
		\biguplus_{i \in I_{\overline x}} 
		\biggl\{ 
			\frac{d_j}{v_i} \in \mset A
			  \biggm|   \delta^{-1}(v_i\cdot \underline a)
							\le j 
							\le \delta^{-1}(v_i \cdot \overline a)
		\biggr\} \\	
  &\wwrel=
		\biguplus_{i\in I_{\overline x}} 
		\biggl\{ 
			\frac{d_j}{v_i} \in \mset A
			  \biggm|   v_i\cdot \underline a
							\le d_j 
							\le v_i \cdot \overline a
		\biggr\} \\
	&\wwrel=
		\biguplus_{i\in I_{\overline x}} 
		\biggl\{ 
			\frac{d_j}{v_i} \in \mset A
			  \biggm| \underline a \le \frac{d_j}{v_i} \le \overline a 
		\biggr\} \\
	&\wwrel=
		\mset A \cap [\underline a, \overline a].
\end{align*}
In particular, the last step follows from \wref{eq:one-sided-candidate-cutoff}
with $\overline{x} > a^*$.
By \wref{lem:candidate-set-linear-corridor}, we know that 
$\underline a \le a^* \le \overline a$ for the
bounds computed in \wref{step:linear-selection-compute-a-bounds},
so we get in particular that $a^* \in \hat{\mset A}$.

It remains to show that we calculate $\hat k$ correctly. Clearly,
we discard with $(a_{i,0}, \ldots, a_{i,\underline j - 1})$ exactly
$\underline{j}$ elements in \wref{step:linear-selection-enlarge-A-hat},
that is $|A_i \cap (-\infty,\underline a)| = \underline{j}(i)$. Therefore,
we compute with
  \[ 
    \hat k 
	= 
		k - \sum_{i\in I_{\overline x}} \bigl| A_i \cap (-\infty,\underline a)\bigr|  
	=
		r(a^*,\mset A) - \bigl|\mset A \cap (-\infty,\underline a)\bigr|
	=  
		r(a^*,\hat{\mset A}) \]
the correct rank of $a^*$ in $\hat{\mset{A}}$.

For the running time, we observe that the computations in 
steps~\ref{step:linear-selection-find-v1} to \ref{step:linear-selection-resultinit} 
are easily done with $\Oh(n)$ primitive instructions.
The loop in \wref{step:linear-selection-compute-A-hat-and-k-hat-loop} and therewith
steps~\ref{step:linear-selection-compute-jmin-jmax} 
and~\ref{step:linear-selection-adapt-k-hat} are executed $|I_{\overline x}| \leq n$ 
times.
The overall number of set operations in \wref{step:linear-selection-enlarge-A-hat}
is $|\hat{\mset A}| \in \Oh(|I_{\overline x}|) \subseteq \Oh(n)$
(cf.\ \wref{lem:candidate-set-linear-corridor}).
Finally, \wref{step:linear-selection-select-A-hat} runs in time
$\Oh(|\hat{\mset A}|) \subseteq \Oh(n)$
when using a (worst"=case) linear"=time rank selection algorithm
(e.\,g., the median"=of"=medians algorithm~\cite{Blum1973}).
\end{proof}
We have obtained a relatively simple algorithm that implements many
divisor methods and has optimal asymptotic running time in the worst case.
It remains to be seen if it is also efficient in practice. 

%% file: app-comparison.tex
\section{Comparison of Algorithms}
\label{sec:algorithms-comparison}

We have implemented all algorithms mentioned above in Java \cite{OurCode}
with a focus on clarity and performance.
Reviewing the algorithms resp.\ implementations (cf.\ \wref{app:algorithm-review}), 
we observe that neither \HAalg nor \Pukalg are asymptotically worst"=case 
efficient whereas \CEalg does not seem to be practical regarding implementability. 
\RWalg does not have either deficiency and is still the shortest of the
non"=trivial algorithms.

We evaluate relative practical efficiency by performing running time experiments
on artificial instances; we fix the number of parties~$n$, house size~$k$ and
the used divisor method and draw multiple vote vectors~\vect v at random according
to different distributions. Where possible, we draw votes from a \emph{continuous}
distribution with fixed expectation; this ensures that vote proportions do not
devolve to trivial situations as $n$ grows.

In order to keep the parameter space manageable, we use $n$ as free variable and 
fix $k$ to a multiple of $n$.
For ease of implementation, we restrict ourselves to divisor sequences of the
form $(\alpha j + \beta)_{j \in \N_0}$; this still allows us to cover a range 
of relevant divisor methods at least approximately (cf.\ \wref{tab:divisor-sequences}).
We describe the machine configuration used for the experiments and 
further details of the setup in \wref{app:experiments-machine-spec}.

\wref{fig:running-times-by-n} shows the results of two experiments with practical
parameter choices. It is clear that \Pukalg dominates the field;
of the other algorithms, only \RWalg comes close in performance. 
These observations are stable across many parameter choices; see also
\wref{app:moreplots}.
We will therefore restrict ourselves to \Pukalg and \RWalg in the sequel.
\begin{figure}
	\plaincenter{%
	  \begin{adjustbox}{width=0.9\linewidth}
      \tikzname{runtime_plots_practical}
      \input{plots/runtime_plots_practical.tikz}
	  \end{adjustbox}
	}
	\caption{%
		This figure shows average running times of 
		  \protect\RWalg\,\protect\plotref{plot:practical_SelectAstar}, 
  		\protect\CEalg\,\protect\plotref{plot:practical_AStarChengEppstein},
  		\protect\Pukalg with naive~\protect\plotref{plot:practical_PukelsheimLS} resp.\
  		                   priority-queue~\protect\plotref{plot:practical_PukelsheimPQ}
  		                   minimum selection, and
  		\protect\HAalg with naive~\protect\plotref{plot:practical_IterativeDMLS} resp.\  
  		                   priority-queue~\protect\plotref{plot:practical_IterativeDMPQ} 
  		                   minimum selection,
  	normalized by the number of parties~$n$. 
		The inputs are random apportionment instances with vote counts $v_i$
		drawn i.\,i.\,d.\ uniformly from $[1,3]$.
		The numbers of parties~$n$, house size~$k$ and method parameters~$(\alpha, \beta)$
		have been chosen to resemble national parliaments in Europe (left)
		and the U.\,S.\ House of Representatives (right), respectively.
	}
	\label{fig:running-times-by-n}
\end{figure}

Towards understanding what influences the performance of these algorithms the most,
we have investigated how $\Delta_a$ (the number of seats \Pukalg assigns too much, 
i.\,e.\ $k - \sum s_i$) resp.\ $|\hat{\mset{A}}|$ (the number of candidates \RWalg 
selects from) relate to the measured running times. While the connection is clear 
for \RWalg, we need to look at cases where \citeauthor{Pukelsheim2014}'s estimators
are bad; as long as $|\Delta_a| \ll n$, the $\Th(n)$ portions of \Pukalg
dominate. \wref{fig:running-times-by-counter} exhibits such a setting.
\begin{figure}
  \plaincenter{%
    \begin{adjustbox}{width=0.9\linewidth}
      \tikzname{scatter_plots_large_exponential}
      \input{plots/scatter_plots_large_exponential.tikz}
	  \end{adjustbox}
  }
  \caption{%
    Running times on individual inputs plotted against $|\Delta_a|$ for \Pukalg (left)
    resp.\ $|\hat{\mset{A}}|$ for \RWalg (right).
    Inputs are random with exponentially distributed~$v_i$
    for $n \in \{%
      1\,\protect\plotref{plot:scatter_PukelsheimPQ_1000},
      5\,\protect\plotref{plot:scatter_PukelsheimPQ_5000},
      10\,\protect\plotref{plot:scatter_PukelsheimPQ_10000},
      20\,\protect\plotref{plot:scatter_PukelsheimPQ_20000},
      30\,\protect\plotref{plot:scatter_PukelsheimPQ_30000},
      40\,\protect\plotref{plot:scatter_PukelsheimPQ_40000},
      50\,\protect\plotref{plot:scatter_PukelsheimPQ_50000},
      75\,\protect\plotref{plot:scatter_PukelsheimPQ_75000},
      100\,\protect\plotref{plot:scatter_PukelsheimPQ_100000} 
    \} \cdot 10^3$ and $k=5n$; they have been apportioned w.\,r.\,t.\ $(\alpha, \beta) = (2,1)$.
  }
  \label{fig:running-times-by-counter}
\end{figure} 

While \Pukalg is faster than \RWalg in the experiments of \wref{fig:running-times-by-n}
and similar ones, we observe that \RWalg is more robust against changing parameters.
\wref{fig:robustness-check} exhibits this for switching between different vote
distributions: the average running times of \RWalg are close to each other
where those of \Pukalg spread out quite a bit. It may be noteworthy that each
algorithm has one ``outlier'' distribution but they are not the same.
\begin{figure}
  \plaincenter{%
    \begin{adjustbox}{width=0.9\linewidth}
      \tikzname{robustness_plots}
      \input{plots/robustness_plots.tikz}
	  \end{adjustbox}
  }
  \caption{%
    Normalized average runtimes of \RWalg (left) and \Pukalg (right)
    on $v_i$ drawn randomly from
      uniform~\protect\plotref{plot:robustness_uniform},
      exponential~\protect\plotref{plot:robustness_exp},
      Poisson~\protect\plotref{plot:robustness_poisson} and
      Pareto~\protect\plotref{plot:robustness_pareto}
    distributions, respectively, and
    with $k=5n$ and $(\alpha, \beta) = (2,1)$.
  }
  \label{fig:robustness-check}
\end{figure}

\Pukalg does indeed seem to outperform \RWalg consistently so far, if not by much
in some cases.
We \emph{have} found a parameterization which, even though it is admittedly
rather artificial, clearly suggests that \Pukalg does indeed have 
$\omega(n)$ worst"=case behavior and that \RWalg can be faster; 
see \wref{fig:pukelsheim-bad-case}. The question after realistic settings for
which this is the case remains open.
\begin{figure}
  \plaincenter{%
    \begin{adjustbox}{width=0.9\linewidth}
      \tikzname{pukelsheim_bad_case}
      \input{plots/pukelsheim_bad_case.tikz}
	  \end{adjustbox}
  }
  \caption{%
    The left plot shows normalized running times of 
      \RWalg\,\protect\plotref{plot:pukelbadcase_SandwichSelect}
    and
      \Pukalg\,\protect\plotref{plot:pukelbadcase_PukelsheimPQ}
    on instances with $k=2n$ and Pareto-distributed $v_i$ for 
    $(\alpha, \beta) = (1.0, 0.001)$.
    The right plot shows that the average of $|\Delta_a|$ seems to converge 
    towards a constant fraction of $n$ in this case.
  }
  \label{fig:pukelsheim-bad-case}
\end{figure}
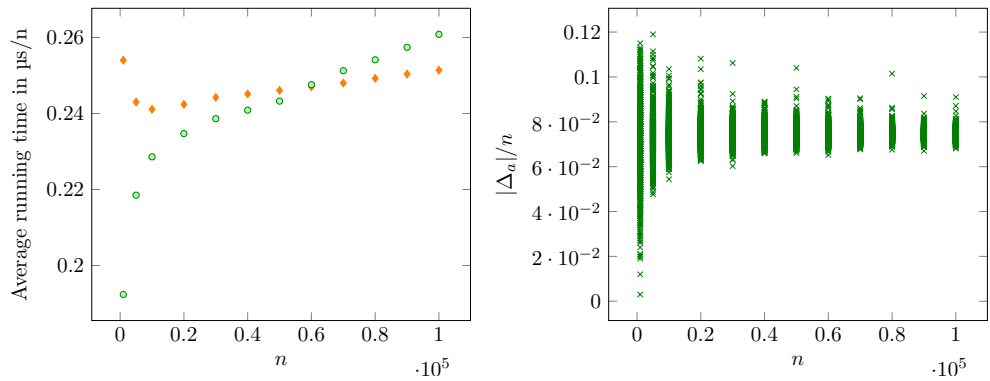

In summary, we have seen that \RWalg provides good performance in a reliable
way, i.\,e., its efficiency does not depend much on divisor sequence or 
input. On the other hand, \Pukalg is faster on average when good estimators
are available, but can be slower in certain settings.

%% file: plots/runtime_plots_practical.tikz
\begin{tikzpicture}
  \pgfplotsset{every axis/.append style={
    cycle list={
      {orange,only marks,mark=diamond*,mark options={scale=1}},         
      {green!50!black,only marks,mark=*,mark options={fill=green!35,scale=0.8}}, 
      {green!50!black,only marks,mark=square*,mark options={fill=green!35,scale=0.8}}, 
      {red,only marks,mark=o,mark options={scale=1}},       
      {black,only marks,mark=asterisk,scale=1},                               
      {black,only marks,mark=x,scale=1}                                 
    }
  }}

  \begin{axis}[xlabel={$n$},ylabel=Average running time in \protect\si{\us \per n},
               title={$(\alpha,\beta) = (2,1)$ and $k = 100n$},
               name=plot1]
    \foreach \alg in %
      {SelectAstar,PukelsheimPQ,PukelsheimLS,AStarChengEppstein,%
       IterativeDMPQ,IterativeDMLS} {%
      
      \IfFileExists{plots/\alg_practical_europe.tab}{%
        \addplot 
          table[x=n,y expr=\thisrow{avg-run-ms/n}*1000] 
          {plots/\alg_practical_europe.tab};
        \label{plot:practical_\alg}%
      }{%
        \PackageError{runtime_plots_practical.tikz}{Plot data missing.
          We need file plots/\alg_practical_europe.tab.
          Remember to run script separate_by_alg
        }{}%
      }
    }
  \end{axis}
  
  \begin{axis}[xlabel={$n$},
               title={$(\alpha,\beta) = (1,\nicefrac{3}{4})$ and $k = 10n$},
               name=plot2,at={($(plot1.east)+(1.5cm,0)$)},anchor=west]
    \foreach \alg in %
      {SelectAstar,PukelsheimPQ,PukelsheimLS,AStarChengEppstein} {%
      
      \IfFileExists{plots/\alg_practical_ushor.tab}{%
        \addplot 
          table[x=n,y expr=\thisrow{avg-run-ms/n}*1000] 
          {plots/\alg_practical_ushor.tab};
      }{%
        \PackageError{runtime_plots_practical.tikz}{Plot data missing.
          We need file plots/\alg_practical_ushor.tab.
          Remember to run script separate_by_alg
        }{}%
      }
      
   }
  \end{axis}

\end{tikzpicture}

%% file: plots/scatter_plots_large_exponential.tikz
\begin{tikzpicture}
  \pgfplotsset{every axis/.append style={cycle list={
      {purple!100!blue,only marks,mark=x},
      {purple!80!blue,only marks,mark=o},
      {purple!60!blue,only marks,mark=x},
      {purple!40!blue,only marks,mark=o},
      {purple!20!blue,only marks,mark=x},
      {blue,only marks,mark=o},
      {blue!75!orange,only marks,mark=x},
      {blue!50!orange,only marks,mark=o},
      {blue!25!orange,only marks,mark=x},
      {orange,only marks,mark=o}
    }}
  }

  \begin{axis}[xlabel=$|\Delta_a|$,ylabel=Running time in \si{\ms},
               name=plot1]
    \foreach \n in {1000,5000,10000,20000,30000,40000,50000,75000,100000} {%
      
      \IfFileExists{plots/PukelsheimPQ_single_large_exponential_\n.tab}{%
        \addplot
          table[x expr=abs(\thisrowno{14}),y=single-run-ms] 
          {plots/PukelsheimPQ_single_large_exponential_\n.tab};
        \label{plot:scatter_PukelsheimPQ_\n}%
      }{%
        \PackageError{scatter_plots_large_exponential.tikz}{Plot data missing.
          We need file plots/PukelsheimPQ_single_large_exponential_\n.tab.
          Remember to run scripts separate_by_alg and separate_by_n
        }{}%
      }
    }
  \end{axis}
  
  \begin{axis}[xlabel=$|\hat{\mset{A}}|$,
               name=plot2,at={($(plot1.east)+(1.5cm,0)$)},anchor=west]
    \foreach \n in {1000,5000,10000,20000,30000,40000,50000,75000,100000} {%
      
      \IfFileExists{plots/SelectAstar_single_large_exponential_\n.tab}{%
        \addplot
          table[x index=16,y=single-run-ms] 
          {plots/SelectAstar_single_large_exponential_\n.tab};
      }{%
        \PackageError{scatter_plots_large_exponential.tikz}{Plot data missing.
          We need file plots/SelectAstar_single_large_exponential_\n.tab.
          Remember to run scripts separate_by_alg and separate_by_n
        }{}%
      }
    }
  \end{axis}

\end{tikzpicture}

%% file: plots/robustness_plots.tikz
\begin{tikzpicture}
  \pgfplotsset{every axis/.append style={
    cycle list={
      {orange,only marks,mark=o,mark options={scale=0.8}},         
      {green!50!black,only marks,mark=x,mark options={scale=0.8}}, 
      {blue!75!black,only marks,mark=square,mark options={scale=0.8}}, 
      {red!25!black,only marks,mark=+,mark options={scale=0.8}}
    }
  }}

  \begin{axis}[xlabel={$n$},ylabel=Average running time in \protect\si{\us \per n},
               ymin=0.045,ymax=0.115,
               name=plot1]
    \foreach \dist in {uniform,exp,poisson,pareto} {%
      \IfFileExists{plots/SelectAstar_robustness_\dist.tab}{%
        \addplot 
          table[x=n,y expr=\thisrow{avg-run-ms/n}*1000] 
          {plots/SelectAstar_robustness_\dist.tab};
        \label{plot:robustness_\dist}%
      }{%
        \PackageError{robustness_plots.tikz}{Plot data missing.
          We need file plots/SelectAstar_robustness_\dist.tab.
          Remember to run script separate_by_alg
        }{}%
      }
    }
  \end{axis}
  
  \begin{axis}[xlabel={$n$},
               ymin=0.035,ymax=0.115,
               name=plot2,at={($(plot1.east)+(2cm,0)$)},anchor=west]
    \foreach \dist in {uniform,exp,poisson,pareto} {%
      \IfFileExists{plots/PukelsheimPQ_robustness_\dist.tab}{%
        \addplot 
          table[x=n,y expr=\thisrow{avg-run-ms/n}*1000] 
          {plots/PukelsheimPQ_robustness_\dist.tab};
      }{%
        \PackageError{robustness_plots.tikz}{Plot data missing.
          We need file plots/PukelsheimPQ_robustness_\dist.tab.
          Remember to run script separate_by_alg
        }{}%
      }
    }
  \end{axis}

\end{tikzpicture}

%% file: plots/pukelsheim_bad_case.tikz
\begin{tikzpicture}
  \pgfplotsset{every axis/.append style={
    cycle list={
      {orange,only marks,mark=diamond*,mark options={scale=1}},         
      {green!50!black,only marks,mark=*,mark options={fill=green!35,scale=0.8}}, 
    }
  }}

  \begin{axis}[xlabel={$n$},ylabel=Average running time in \protect\si{\us \per n},
               name=plot1]
    \foreach \alg in %
      {SandwichSelect,PukelsheimPQ} {%
      
      \IfFileExists{plots/\alg_pukelsheim_bad_case_avg.tab}{%
        \addplot 
          table[x=n,y expr=\thisrow{avg-run-ms/n}*1000] 
          {plots/\alg_pukelsheim_bad_case_avg.tab};
        \label{plot:pukelbadcase_\alg}%
      }{%
        \PackageError{pukelsheim_bad_case.tikz}{Plot data missing.
          We need file plots/\alg_pukelsheim_bad_case_avg.tab.
          Remember to run script separate_by_alg
        }{}%
      }
    }
  \end{axis}
  
  \begin{axis}[xlabel={$n$},ylabel={$|\Delta_a|/n$},
               name=plot2,at={($(plot1.east)+(2.5cm,0)$)},anchor=west]
    \foreach \alg in %
      {PukelsheimPQ} {%
      
      \IfFileExists{plots/\alg_pukelsheim_bad_case_trimmed.tab}{%
        \addplot[only marks,mark=x,green!50!black]
          table[x=n,y expr=abs(\thisrow{counter/n})] 
          {plots/\alg_pukelsheim_bad_case_trimmed.tab};
        \label{plot:pukelbadcase_\alg_scatter}%
      }{%
        \PackageError{pukelsheim_bad_case.tikz}{Plot data missing.
          We need file plots/\alg_pukelsheim_bad_case_trimmed.tab.
          Remember to run script separate_by_alg
        }{}%
      }
      
   }
  \end{axis}

\end{tikzpicture}

%% file: app-methods-scope.tex
\section{Our Scope of different Methods of Apportionment}%
\label{app:methods-scope}

As we have seen in \wref{sec:notation} there are many possible divisor
sequences.
For our main result (cf.\ page~\pageref{lem:candidate-set-linear-corridor})
we follow \textcite{Cheng2014} and require the sequences to be ``almost''
linear;
we should check that we do not unduly restrict the scope of our investigation.

We refer to the recent reference work by \textcite{Pukelsheim2014} and, by
extension, to \textcite{BalinskiYoung2001}
who classify different divisor methods of apportionment in terms of
\emph{sign"-post} sequences, a concept equivalent to the divisor sequences
we use.
They distinguish these classes of such sequences
(cf.\ \cite[Sections~3.11-12]{Pukelsheim2014}):
\begin{itemize}
  \item \emph{stationary} sign"=posts of the form $s(n) = n-1+r$
    with $r \in (0,1)$;
  \item \emph{power"=mean} sign"=posts defined by
    \begin{align*}
      \tilde{s}_p(0) &= 0, \\
      \tilde{s}_p(n) &=
        \biggl( \frac{(n-1)^p + n^p}{2} \biggr)^{\nicefrac{1}{p}},
    \end{align*}
    for $p \neq -\infty, 0, \infty$;
  \item and special cases
    $\tilde{s}_{-\infty}(n) = n-1$,
    $\tilde{s}_0(n) = \sqrt{(n-1)n}$, and
    $\tilde{s}_{\infty}(n) = n$.
\end{itemize}
It is easy to see that stationary sign"=posts correspond do divisor sequences
$d_j = j + \beta$ with $\beta \in (0,1)$ (up to a shift by one);
as such, \wref{lem:candidate-set-linear-corridor} applies immediateley with
$\alpha = 1$ and $\underline{\beta} = \overline{\beta} = \beta$, and yields a
particularly nice
(and tight, for our choices of $\underline{a}$ and $\overline{a}$)
upper bound on the size of the candidate set $\mset{A}$.
We cover the special cases as well; see  \wref{tab:divisor-sequences} for
the corresponding sandwich bounds.

As for the remaining power"=mean sign"=posts, the trivial bounds
$\underline{\beta} = 0$ and $\overline{\beta} = 1$ already work.
One can apply the power"=mean inequality and use the slightly better bounds for
$p \in \{-\infty, -1,0,1, \infty\}$ as given in \wref{tab:divisor-sequences}.
Even better bounds can be gleaned from observing that
$\tilde{s}_p(n)$ converges to $n - \nicefrac{1}{2}$ from one side, and quickly
so; $\tilde{s}_p(1)$ thus determines either $\underline{\beta}$ or
$\overline{\beta}$ and the other can be chosen as $\nicefrac{1}{2}$.

In summary, our algorithm \RWalg applies to all divisor methods treated by
\textcite{Pukelsheim2014} and \textcite{BalinskiYoung2001}

%% file: app-proofs.tex
\section{Lemmata and Proofs}%
\label{app:lemmata}

\begin{lemma}\label{lem:rank-as-sum}
	For rank function $r(x,\mset A)$,
	\[
			r(x,\mset A)
		\wwrel=
			\sum_{i=1}^n \lfloor \delta^{-1}(v_i \cdot x) \rfloor + 1.
	\]
	Moreover, for $x < \overline x$ we have
	\[
			r(x,\mset A)
		\wwrel=
			\sum_{\mathclap{i \in I_{\overline x}}} 
				\lfloor \delta^{-1}(v_i \cdot x) \rfloor + 1
	\]
	with 
	  $I_{\overline x} = \bigl\{i \in \{1,\dots,n\} \mid v_i > d_0/\overline x \bigr\}$.
\end{lemma}

  \subsection[Proof of Lemma 4]{Proof of \wref{lem:rank-as-sum}}%

  \label{app:proof-rank-as-sum}%
  
	By \wildtpageref[eq.\!]{eq:rank-as-sum}\eqref, it suffices to show that
	  \[  \bigl| \{a_{i,j} \mid a_{i,j} \le x\} \bigr| 
	    = \lfloor \delta^{-1}(v_i \cdot x) \rfloor + 1 \] 
	for each $i \in \{1, \dots, n\}$.
	Now, if $x \ge a_{i,j} = \nicefrac{d_j}{v_i}$ for some $j$,
	then $v_i \cdot x \ge d_j$, and	so $\lfloor \delta^{-1}(v_i\cdot x) \rfloor$ 
	is the largest index $j'$ for which	$a_{i,j'} = d_{j'}/v_i \le x$. 
	As $d_j$ is zero-based, there are $j'+1 \geq 1$ such elements $a_{i,j} \le x$ 
	and the equation follows.
	
	Otherwise, that is $a_{i,j} > x$ for all $j$, we have 
	$j' = \lfloor \delta^{-1}(v_i\cdot x) \rfloor = -1$ 
	by~\ref{item:delta-inverse-below-d0} and \wref{cor:delta-inverse-monotonicity}
	and the equality holds with $0$ on both sides.
	
	For the second equality, we only have to show that the omitted summands 
	are zero. So let $i \notin I_{\overline x}$ be given, that is 
	$v_i \le d_0/\overline x$. For $x < \overline x$, we have
	\[
			v_i \cdot x 
		\wwrel\le
			\frac{d_0}{\overline x} \cdot x
		\wwrel< 
			\frac{d_0}{\overline x} \cdot \overline x 
		\wwrel= 
			d_0,
	\]
	and hence $\lfloor \delta^{-1}(v_i \cdot x) \rfloor = -1$ 
	by \ref{item:delta-inverse-below-d0}. 

\begin{lemma}
  \label{lem:astar-corridor}
	Let $\overline x > a^*$ and
	assume $\overline a$ and $\underline a$ are chosen so that they
	fulfill
	\[
		\sum_{i \in I_{\overline x}} 
		\delta^{-1}(v_i \cdot \underline a) \wwrel\le k - |I_{\overline x}|
	\qquad\text{and}\qquad
		\sum_{i \in I_{\overline x}} \delta^{-1}(v_i\cdot\overline a)
		\wwrel\ge k.
	\]
	Then, $\underline a \le a^* \le \overline a$.
\end{lemma}
The lemma follows more or less directly; one uses the sandwich bounds on $r$
to show that $a < \underline a$ are infeasible, i.\,e., $r(a) < k$, and
that $\overline a$ is feasible, and thus all $a > \overline a$ are suboptimal
since $a^*$ is the smallest feasible element in $\mset{A}$.

  \subsection[Proof of Lemma 5]{Proof of \wref{lem:astar-corridor}}%

  \label{app:proof-astar-corridor}%

  As a direct consequence of \wref{lem:rank-as-sum} together with 
  the fundamental bounds $y-1 < \lfloor y \rfloor \leq y$ on floors,
  we find that
  \begin{equation}
  \label{eq:corridor-bounds}
		  \sum_{\mathclap{i \in I_{\overline x}}} 
			  \delta^{-1}(v_i \cdot x)
	  \wwrel<
		  r(x, \mset A)
	  \wwrel\le
		  \sum_{\mathclap{i \in I_{\overline x}}}
			  \bigl( \delta^{-1}(v_i \cdot x) + 1 \bigr)
	  \wwrel=
		  |I_{\overline x}|
		  \bin+ \sum_{\mathclap{i \in I_{\overline x}}}
			  \delta^{-1}(v_i \cdot x)
  \end{equation}
  for any $\overline x$ and all $x < \overline x$.
  We can therewith pin down the value of $r$ to an interval of 
  width $|I_{\overline x}|$ using only $\delta^{-1}$.
  We can use this to derive upper \emph{and} lower bounds on $a^*$.

  We show that smaller $a$ are infeasible and larger $a$ are clearly suboptimal,
  so the optimal $a^*$ must lie in between.
  Let us first consider $a<\underline a$. 
  There are two cases: if there is a $v_i$, such that $v_i a \ge d_0$, 
  we get by strict monotonicity of $\delta^{-1}$
  \begin{align*}
	    r(a)
    &\wwrel{\relwithref{eq:corridor-bounds}{\le}}
	    |I_{\overline x}|
	    \bin+ \sum_{\mathclap{i \in I_{\overline x}}}
		    \delta^{-1}(v_i \cdot a)
  \\	&\wwrel<
	    |I_{\overline x}|
	    \bin+ \sum_{\mathclap{i \in I_{\overline x}}}
		    \delta^{-1}(v_i \cdot \underline a)
  \\	&\wwrel\le
	    k
  \end{align*}
  and $a$ is infeasible.
  If otherwise $v_i a < d_0$, i.\,e., $a < d_0/v_i$, for all $i$,
  $a$ must clearly have rank $r(a) = 0$ as it is smaller than any element
  $a_{i,j}\in\mset A$.
  In both cases we found that $a<\underline a$ has rank $r(a) < k$.

  Now consider the upper bound, i.\,e., we have $a > \overline a$. 
  In case $\overline a \ge \overline x$, 
  we have $a > \overline x > a^*$ by assumption and any such $a$ 
  cannot be optimal.
  Otherwise, for $\overline a < \overline x$,
  we have
  \begin{align*}
	    r(\overline a)
    &\wwrel{\relwithref{eq:corridor-bounds}{>}}
	    \sum_{\mathclap{i \in I_{\overline x}}} 
			    \delta^{-1}(v_i \cdot \overline a)
    \wwrel\ge
	    k,
  \end{align*}
  so $\overline a$ is feasible. 
  Any element $a > \overline a$ can thus not be the optimal solution~$a^*$,
  which is the \emph{minimal} $a$ with $r(a) \ge k$.

  \subsection[Proof of Lemma 2]{Proof of \wref{lem:candidate-set-linear-corridor}}%

  \label{app:proof-candidate-set-linear-corridor}%
  
  We consider the linear divisor sequence continuations 
    \[ \underline\delta(j) = \alpha j + \underline\beta 
         \qquad\text{and}\qquad
       \overline\delta(j) = \alpha j + \overline\beta \]
  for all $j\in\R_{\ge0}$ and start by noting that the inverses are
    \[ \underline\delta{}^{-1}(x) = \nicefrac x\alpha - \nicefrac{\underline\beta}\alpha
         \qquad\text{and}\qquad
       \overline\delta{}^{-1}(x) = \nicefrac x\alpha - \nicefrac{\overline\beta}\alpha \]
  for 
    $x\ge \underline\delta(0) = \underline\beta$ 
  and 
    $x\ge \overline\delta(0) = \overline\beta$,
  respectively.
  For smaller $x$, we are free to choose the value of the continuation from $[-1,0)$ 
  (cf.\ \ref{item:delta-inverse-below-d0});
  noting that $\nicefrac{x}{\alpha} - \nicefrac{\overline\beta}{\alpha} < 0$ for 
  $x < \overline\beta$, a choice that will turn out convenient is
  \begin{equation}
  \label{eq:delta-lin-inverse-as-max}
		  \underline\delta{}^{-1}(x) 
	  \wrel\ce 
		  \max \biggl\{\frac x\alpha - \frac{\underline\beta}\alpha,\; -1 \biggr\}
	  \quad\text{resp.}\quad
		  \overline\delta{}^{-1}(x) 
	  \wrel\ce 
		  \max \biggl\{\frac x\alpha - \frac{\overline\beta}\alpha,\; -1 \biggr\}.
  \end{equation}
  We state the following simple property for reference; 
  it follows from $\underline \delta(j) \le \delta(j) \le \overline\delta(j)$ 
  and the definition of the inverses (recall that $\underline\beta \leq \alpha$):
  \begin{equation}
  \label{eq:corridor-delta-lin-corridor-inverse}
		  \frac {x}\alpha - \frac{\overline\beta}\alpha
	  \wwrel\le
		  \overline\delta{}^{-1}(x)
	  \wwrel\le
		  \delta^{-1}(x)
	  \wwrel\le
		  \underline\delta{}^{-1}(x)
	  \wwrel\le
		  \frac {x}\alpha - \frac{\underline\beta}\alpha,
	  \qquad\text{for } x \ge 0.
  \end{equation}
  Equipped with these preliminaries, we compute
  \begin{align*}
		  &  \overline a
		  \wwrel= 
		  \frac{\alpha k + \overline\beta|I_{\overline x}|}{V_{\overline x}}.
  \\\wwrel\iff
		  &  \frac{\overline a}\alpha \cdot
			  \sum_{i \in I_{\overline x}} v_i
		  \wwrel= 
		  k \bin+ \frac{\overline\beta}{\alpha} \cdot |I_{\overline x}| ,
  \\\wwrel\iff
		  &  k
		  \wwrel=
			  \sum_{i \in I_{\overline x}} \biggl(
				  \frac{v_i \cdot \overline a}{\alpha} - \frac{\overline\beta}{\alpha}
			  \biggr)
		  \wwrel{\relwithref{eq:corridor-delta-lin-corridor-inverse}\le}
			  \sum_{i \in I_{\overline x}} \delta^{-1}(v_i\cdot\overline a),
  \end{align*}
  so $\overline a$ satisfies the condition of \wref{lem:astar-corridor}.
  Similarly, we find 
  \begin{align*}
		  &  \underline a
		  \wwrel= 
			  \frac{\alpha k - (\alpha - \underline\beta) \cdot |I_{\overline x}|}%
			       {V_{\overline x}},
  \\\wwrel\iff
		  &  \frac{\underline a }{\alpha} \cdot V_{\overline x}
		  \wwrel= 
		  k - (1 - \nicefrac{\underline\beta}{\alpha}) \cdot |I_{\overline x}|,
  \\\wwrel\iff
		  &	k
		  \wwrel=
			  |I_{\overline x}| + \sum_{i \in I_{\overline x}} 
							  \biggl(   \frac{v_i \cdot \underline a}{\alpha}
							          - \frac{\underline\beta}{\alpha} \biggr)
		  \wwrel{\relwithref{eq:corridor-delta-lin-corridor-inverse}\ge}
			   |I_{\overline x}| + \sum_{i \in I_{\overline x}} 
			  \delta^{-1}(v_i \cdot \underline a),
  \end{align*}
  that is $\underline a$ also fulfills the conditions of \wref{lem:astar-corridor}.

  \medskip\noindent
  For the bound on the number of elements falling between $\underline a$ and $\overline a$,
  we compute
  \begin{align*}
		  \bigl|\mset A \cap [\underline a,\overline a] \bigr|
	  &\wwrel=
		  \sum_{i\in I_{\overline x}} \bigl|A_i \cap [\underline a,\overline a] \bigr|
  \\	&\wwrel=
		  \sum_{i\in I_{\overline x}} \Biggl|\biggl\{ 
			  j\in\N_0 \biggm| \underline a \le \frac{d_j}{v_i} \le \overline a 
		  \biggr\} \Biggr|
  \\	&\wwrel=
		  \sum_{i\in I_{\overline x}} \Bigl|\bigl\{ 
			  j\in\N_0 \bigm| 
				  v_i\cdot \underline a
				  \le d_j 
				  \le v_i \cdot \overline a
		  \bigr\} \Bigr|
  \\	&\wwrel=
		  \sum_{i\in I_{\overline x}} \Bigl|\bigl\{ 
			  j\in\N_0 \bigm| 
				  \delta^{-1}(v_i\cdot \underline a)
				  \le j 
				  \le \delta^{-1}(v_i \cdot \overline a)
		  \bigr\} \Bigr|
  \\	&\wwrel\le
		  \sum_{i\in I_{\overline x}} \Bigl(
			  \delta^{-1}( v_i \cdot \overline a )
			  -\delta^{-1}( v_i \cdot \underline a )
			  +1
		  \Bigr)
  \\	&\wwrel{\relwithref{eq:corridor-delta-lin-corridor-inverse}\le}
		  \sum_{i\in I_{\overline x}} \Bigl(
			    \underline\delta^{-1}( v_i \cdot \overline a )
			  - \overline\delta^{-1}( v_i \cdot \underline a )
			  +1
		  \Bigr)
  \\	&\wwrel{\relwithref{eq:corridor-delta-lin-corridor-inverse}\le}
		  \sum_{i\in I_{\overline x}} \Bigl(
			    \frac{v_i \cdot \overline a - \underline\beta}{\alpha} 
			  - \frac{v_i \cdot \underline a - \overline\beta}{\alpha}
			  +1 
		  \Bigr)
  \\	&\wwrel=
		  \sum_{i\in I_{\overline x}} \Bigl(
			        1 + \frac{\overline\beta - \underline\beta}{\alpha}
			  \bin+ \frac{v_i\cdot \overline a - v_i \cdot \underline a}\alpha
		  \Bigr)
  \\	&\wwrel=
		  \biggl( 1 + \frac{\overline\beta - \underline\beta}{\alpha}\biggr) 
		    \cdot |I_{\overline x}| 
		  \wbin+ 
		    (\overline a - \underline a) \cdot \frac{{V_{\overline x}}}{\alpha}
  \\	&\wwrel=
		  \biggl( 1 + \frac{\overline\beta - \underline\beta}{\alpha}\biggr) 
		    \cdot |I_{\overline x}| 
		  \wbin+ \frac{(\alpha + \overline\beta - \underline\beta) 
		                 \cdot |I_{\overline x}|}%
		              {V_{\overline x}}
			  \cdot\frac{{V_{\overline x}}}{\alpha}
  \\	&\wwrel=
		  2\biggl(1 + \frac{\overline\beta - \underline\beta}{\alpha}\biggr) 
		    \cdot |I_{\overline x}|.
  \end{align*}

%% file: app-algorithm-reviews.tex
\section{Implementing the Algorithms}%
\label{app:algorithm-review}%

In this section, we review existing algorithms for divisor methods.
In particular, we elaborate on how we have implemented them for our
experiments \cite{OurCode}, and on problems we have encountered in this process.

We have taken care not to render the algorithm unnecessarily inefficient in order
to perform a fair comparison of running times; the result is to the best of our 
abilities conditioned on a limited time budget. In particular, all of our
implementations have been refined on the programming level to 
roughly the same degree.

For the purpose of a fair comparison, all implementation have to conform to
the same interface.
\begin{description}
  \item[Parameters:] A pair $(\alpha, \beta) \in \R^2$ with $\alpha > 0$
    and $\beta > 0$. 
  \item[Input:] Votes $\vect v$ and house size $k$.
  \item[Output:] A (symbolic) representation of all seat assignments valid
    w.\,r.\,t.\ divisor sequence $(\alpha j + \beta)_{j \geq 0}$, as well as
    proportionality constant $a^*$.
\end{description}
More specifically, the output is encoded as a vector of undisputed seats and a 
binary vector indicating which parties are tied for the remaining seats.
We skip the step from $a^*$ resp.\ a valid seat assignment to this representation
in the pseudo code since it is elementary: all parties with ``current'' resp.\ 
``next'' value $v_i / d_{s_i - 1}$ resp.\ $v_i / d_{s_i}$ equal $a^*$ are tied. 
A simple $\Th(n)$"/time post"=processing identifies these in all cases.

We have established confidence in the correctness of our implementations by
extensive random testing \cite[\texttt{TestMain.java}]{OurCode}; every implementation
has been run on thousands of instances. The correctness of the results has been
confirmed, besides rudimentary sanity checks such as matching vector dimensions, 
by checking \citeauthor{Pukelsheim2014}'s 
\emph{Max"=Min Inequality}~\cite[Theorem~4.5]{Pukelsheim2014}.

All implementations share the same numerical weakness, though: using
fixed"=precision arithmetics, two computations that should lead to the same
result (say, $a^*$) yield different numbers. We compensate for that by using
fuzzy comparisons: we identify numbers if they are within some constant $\epsilon$
of each other. Thus, we can reliably identify tied parties, for instance.

There is a drawback, though: if distinct values $v_i / d_j$ are closer than 
$\epsilon$ (or, even without the adaption, the resolution of the chosen
fixed"=precision number representation), we may identify them and thus compute 
wrong seat assignments.

This issue can not be circumvented on the algorithmic level. The only robust resort
we know of is using arbitrary"=precision arithmetics, inevitably slowing down all
the algorithms.


\subsection{Iterative Divisor Method}
\label{app:iterative-method}

Implementing \HAalg is straight"=forward. An implementation using a priority
queue implementation from the standard library runs in time $\Th(n + k \log n)$.
Since we expect overhead for the queue to be significant for small $n$,
we also implement a variant which determines $I$ using a simple linear scan,
resulting in a total running time in $\Th(kn)$.

Shared code aside, \HAalg takes about 50 resp.\ 65 lines of code with resp.\ 
without priority queues.


\subsection{Jump-and-Step}%
\label{app:jump-and-step-review}

The \emph{jump"=and"=step} algorithm~\cite[Section~4.6]{Pukelsheim2014} can be
formulated using our notation as follows:
\begin{algorithm}\label{alg:jump-and-step}
  $\textsc{\PukalgASCII}_d(\vect v, k)$ :
  \begin{enumerate}[label={\textsf{\textbf{Step~\arabic*}}},ref=\arabic*,leftmargin=4.5em]
    \item\label{step:puk-guess}
      Compute an estimator $a$ for $a^*$.
    \item\label{step:init}
      Initialize $s_i = \lfloor \delta^{-1}(v_i \cdot a) \rfloor + 1$.
    \item\label{step:iterate}
      Iterate similarly to \HAalg until $\sum s_i = k$ with
        \[ I = \begin{cases}
               \arg\max_{i=1}^n v_i/d_{s_i},     &\sum s_i < k; \\[1ex]
               \arg\min_{i=1}^n v_i/d_{s_i - 1}, &\sum s_i > k.
             \end{cases} \]
  \end{enumerate}
\end{algorithm}
The performance of this algorithm clearly depends on $\Delta_a \ce \sum s_i - k$
after \wref{step:init};
the running time is in $\Th(n + |\Delta_a| \cdot \log n)$ when using priority queues
for \wref{step:iterate} (which may \emph{not} be advisable in practice 
if $|\Delta_a|$ can be expected to be very small). 
As such, the running time is not per se bounded in $n$ and $k$.

We follow the recommendations of \citeauthor{Pukelsheim2014} and use the
estimator~\cite[Section~6.1]{Pukelsheim2014}
\[ a \wwrel\ce \frac{\alpha}{V} \cdot \begin{cases}
      k + n \cdot (\nicefrac{\beta}{\alpha} - \nicefrac{1}{2}), 
        &0 \leq \nicefrac{\beta}{\alpha} \leq 1 ;\\[1ex]
      k + n \cdot \lfloor \nicefrac{\beta}{\alpha} \rfloor, &\text{else}.
    \end{cases} \]
The first case corresponds to \citeauthor{Pukelsheim2014}'s \emph{recommended}
estimator for \emph{stationary} signpost sequences, the second to
his \emph{good universal} estimator generalized to divisor sequences that are 
not signpost sequences in the strict sense.
The additional factor $\alpha$ rescales the value appropriately; 
\citeauthor{Pukelsheim2014} only considers $\alpha = 1$.

Given that these estimators guarantee $|\Delta_a| \leq n$ in the worst case, we
can assume that \Pukalg runs in time $\Oh(n \log n)$. Furthermore,
\citeauthor{Pukelsheim2014} claims that the recommended estimator is good in
practice in the sense that $|\Delta_a| \in \Oh(1)$ in expectation, so 
\Pukalg may be efficient in practice for large $n$ as well.
Since their proof is limited to uniformly distributed votes and $k \to \infty$, 
we investigate this in \wref{sec:algorithms-comparison}.

Shared code aside, \Pukalg takes about 120 lines of code, with or without 
priority queues.

                          
\subsection[The Algorithm of Cheng and Eppstein]{The Algorithm of \citeauthor{Cheng2014}}
\label{app:CE-main-procedure}

\Textcite{Cheng2014} do not give pseudocode for the main procedure of their algorithm
which would combine the individual steps to compute $\mathcal{A}_{(k)}$.
For the reader's convenience and for clarity concerning our running"=time
comparisons we give this top"=level procedure as we have inferred it.

\begin{algorithm}\label{alg:CE-main-procedure}
  $\textsc{\CEalgASCII}_d(\vect v, k)$ :	
  \begin{enumerate}[label={\textsf{\textbf{Step~\arabic*}}},ref=\arabic*,leftmargin=4.5em]
    \item \label{step:cemain-init}
      Compute a suitable finite representation of $\mathcal{A}$.
    
    \item \label{step:cemain-findcontrseq}
      $\mathcal{C} \ce \textsc{FindContributingSequences}(\mathcal{A}, k)$.
      
    \item \label{step:cemain-findcoarse}
      $\xi \ce s^{-1}(k)$ \cite[(3)]{Cheng2014}.
      
    \item \label{step:cemain-findlowrankcoarse}
      If $r(\xi, \mathcal{A}) \geq k$ then\\
      \hspace*{1em}$\xi \ce \textsc{LowerRankCoarseSolution}(\mathcal{A}, k, \xi)$.
      
    \item \label{step:cemain-returnexact}
      Return $\textsc{CoarseToExact}(\mathcal{A}, k, \xi)$.
  \end{enumerate}
\end{algorithm}
The subroutines are given in sufficient detail in their Algorithms~1 to~3, respectively.
The pseudo code given uses some high"=level set operations which we did not 
implement naively due to performance concerns; we compute several steps
during a single iteration over the respective sets of sequences.

Note that we have (hopefully) fixed an off-by-one mistake in the text.
The definition of rank $r(x,A)$ is, 
  ``the number of elements of $A$ less than or equal to $x$'';
that is, the rank of $A(j)$ is $j+1$ since $A$ is zero"=based (the first element is $A(0)$).
However, the authors continue to say that $r(x,A)$ 
	``is the index $j$ such that $A(j) \le x < A(j+1)$.''
	
Regarding performance, \citeauthor{Cheng2014} show that their algorithm runs in
time $\Th(n)$ in the worst case.
Since \CEalg computes a linear number of medians and requires a linear number of
evaluations of rank function $r(x,\mset A)$ (with geometrically shrinking 
$|\mset A|$~-- otherwise the algorithm would not run in linear time),
it is unclear whether the algorithm is efficient in practice.

Shared code aside, \CEalg take about 300 lines of code. By this measure, it is 
the most complex of the algorithms we consider.
	
\subsection*{Additional Issues with Numerics}
In addition to the concerns expressed above, there are additional 
numerical issues when implementing \CEalg using fixed"=precision floating"=point arithmetics.
In short, we have to compute certain floors and ceilings of real numbers
\emph{exactly} or we may compute a \emph{wrong result}. 

More specifically, we evaluate $r(x,\mset A)$ several times by computing 
terms of the form $\lfloor \delta^{-1}(\_) \rfloor$ (cf.\ \wref{lem:rank-as-sum}).
The problem is that the result of $\delta^{-1}(\_)$ is non"=integral in general, 
but \emph{is} integral when the argument evaluates exactly to a $d_j$.
With the usual floating"=point arithmetic the result might be slightly smaller, though.
We then erroneously round down to the next smaller integer~-- a critical error!

In practice, we can add a small constant to the mantissa before taking the floor.
This constant has to be chosen large enough to cover potential rounding errors,
but also \emph{small enough} so as to not change subsequent calculations; 
\CEalg may compute a wrong answer otherwise. 
This is a very delicate requirement we do not know how to fulfill in general.

	
\subsection{\RWalgASCII}%
\label{app:rwalg-review}%
We already discuss our algorithm at length in \wref{sec:fast-selection-alg}.
Since we want to investigate \emph{practical} performance, we implement
rank"=selection using average"=case efficient Quickselect as opposed to
using a linear"=time algorithm with large constant factors.

We want to emphasize that our final algorithm \RWalg is conceptually simple in 
the sense that there is little hidden complexity.
We need exactly one call to a rank selection algorithm on a linear"=size 
list which takes five additional linear"=time operations to come up with: 
  finding the maximal value $v^{(1)}$,
  constructing index set $I_{\overline x}$, 
  computing $V_{\overline x}$,
  constructing multiset $\hat{\mset A}$
  and computing~$\hat k$.
These are all quite elementary tasks in that they use one \texttt{for}"/loop
each which run for at most $n$ iterations with only few operations in each.
We therefore think that we can outperform \CEalg in practice, and should not be
far behind \Pukalg, either.

Regarding implementation, the delicate part was to get the bounds on $j$ 
(cf.\ \wref{step:linear-selection-compute-jmin-jmax}) right.
We use floor and ceiling functions on real numbers, so rounding errors that occur 
in fixed"=precision floating"=point arithmetic can cause harm.
We can circumvent this by adding (subtracting) a conservatively large constant to the 
mantissa of the floats before taking floors (ceilings).
If this constant is larger than necessary for covering rounding errors, 
we might add slightly more candidates to $\hat{\mset{A}}$ (at most two per party)
which would slightly degrade performance. 
Correctness, however, is \emph{not} affected (in contrast to \CEalg).

We also remark here that the code~\cite{OurCode} for the experimental results discussed in 
\wref{sec:algorithms-comparison}
is based on an earlier version of \wref{lem:candidate-set-linear-corridor}
with slightly weaker bounds (cf.\ \wref{app:errata}).
Experiments with the updated code are to follow, and might yield slight improvements
for \RWalg.

Shared code aside, \RWalg takes about 100 lines of code. By this measure,
it is the least complex of the non"=trivial algorithms we consider.

%% file: app-machine-spec.tex
\section{Experimental Setup}%
\label{app:experiments-machine-spec}%
We have run the experiments with Java~7 on Ubuntu~14.04~LTS 
running kernel 3.13.0-34-generic x86\_64 GNU/Linux.
The hardware platform is a ThinkPad~T430s Tablet with the following core 
parameters according to \texttt{lshw}.
\begin{description}
  \item[CPU:] Intel\textsuperscript{\textregistered} Core{\texttrademark} 
              i5-3320M CPU @ 2.60GHz
  \item[Cache:] L1 32KiB, L2 256KiB, L3 3MiB
  \item[RAM:] 4+4GiB SODIMM DDR3 Synchronous 1600 MHz (0.6 ns)
\end{description}

As our code is written in Java, we include a warm-up phase
to trigger just-in-time compilation of the relevant methods. 
All times are measured using the built-in method \texttt{System.nanoTime()}.
We use the same set of inputs for all algorithms, all of which have to 
construct the full set $\mathcal{S}(\vect v, k)$ for each input $(\vect v, k)$
during the measurement.

In order to increase accuracy, we repeat the execution of each algorithm on each input
several times and measure the total time; we then report the average time per
execution.

For the selection"=based algorithms, we use the randomized Quickselect"=based 
implementation by \textcite{SedgewickWayne2011} as published on the book website. 
We use the (pseudo) random number generators for several distributions 
from the same library (download of \texttt{stdlib-package.jar} on 
August~11th, 2015).

For reproducing our running time experiments, make sure you have
working GNU/Linux\footnote{%
  Our framework \emph{may} work on other platforms, maybe with small adjustments
  to the Ruby code, but we have not tried. See \texttt{README.md} for a workaround.
} installation with Ruby, Java~7 and Ant; then execute
\begin{indented}\ttfamily
	ruby run\_experiments.rb arxiv.experiment
\end{indented}
for the data represented in \wref{sec:algorithms-comparison} and \wref{app:moreplots}.
Be warned: this may run for a long time, and it \emph{will} create lots of
images (provided you have \texttt{gnuplot} installed).

%% file: app-moreplots.tex
\section{More Running-Time Experiments}
\label{app:moreplots}%
We apologize to only offer draft graphics without commentary for the time being.
\IfFileExists{plots/draft/draft_plots_refined}{%
  \input{plots/draft/draft_plots_refined}
}{%
  \IfFileExists{plots/draft/draft_plots}{%
    \input{plots/draft/draft_plots}
  }{%
    \PackageError{app-moreplots.tex}{Draft plots missing.
      Remember to run script draft/to_latex
    }{}%
  }
}

%% file: plots/draft/draft_plots_refined.tex
\clearpage\enlargethispage{\baselineskip}
Average normalized runtimes for several input distributions and across several orders of magnitudes of $n$.

\includegraphics[width=0.45\linewidth]{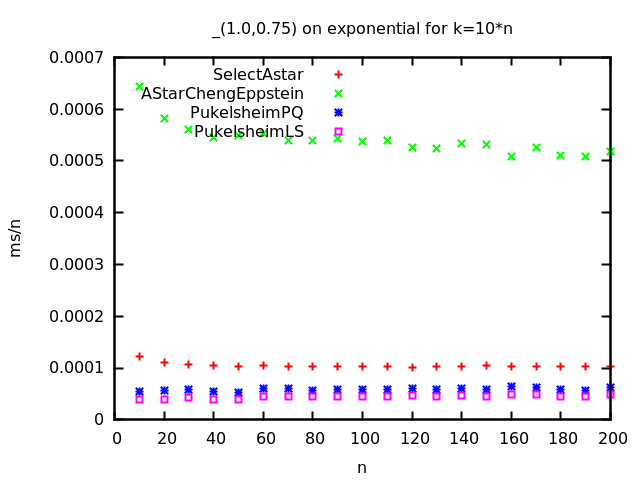}
\hfill\includegraphics[width=0.45\linewidth]{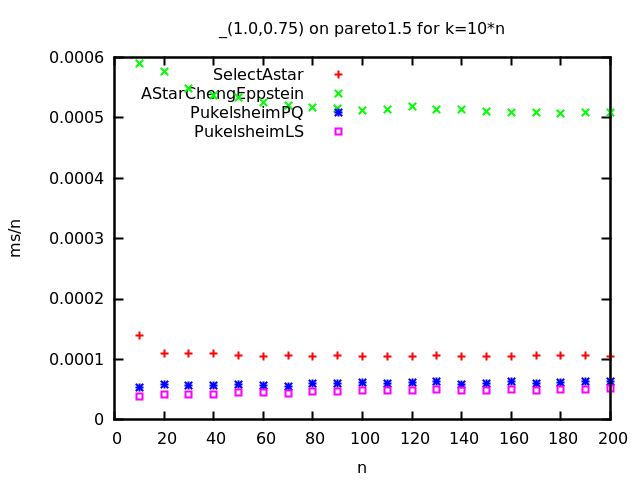}
\hfill\\
\hfill\includegraphics[width=0.45\linewidth]{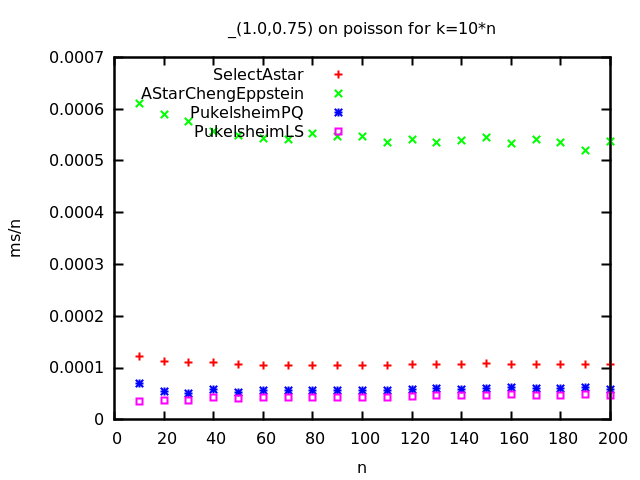}
\hfill\includegraphics[width=0.45\linewidth]{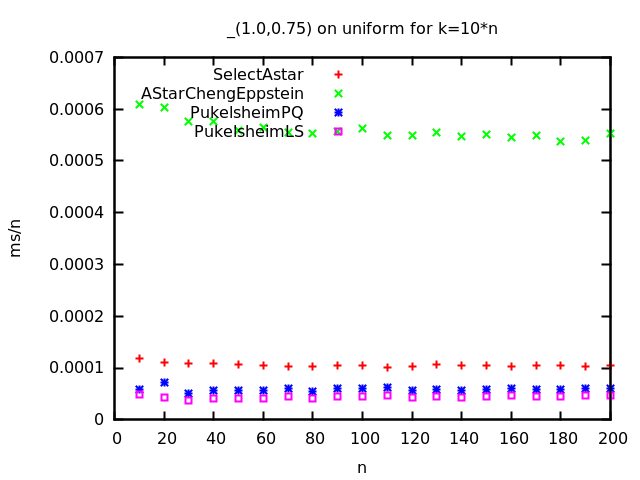}
\hfill\\
\hfill\includegraphics[width=0.45\linewidth]{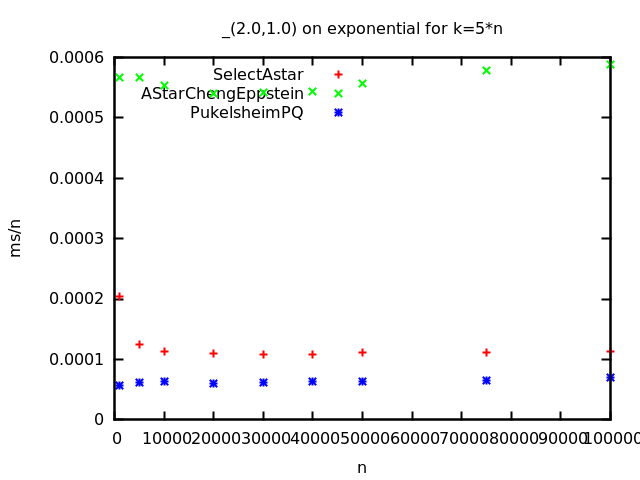}
\hfill\includegraphics[width=0.45\linewidth]{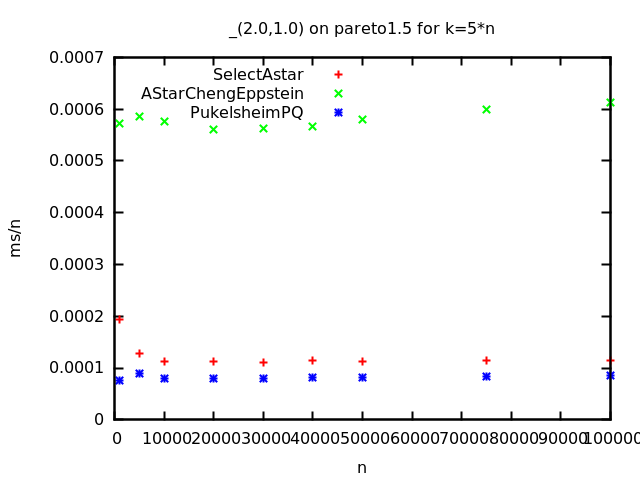}
\hfill\\
\hfill\includegraphics[width=0.45\linewidth]{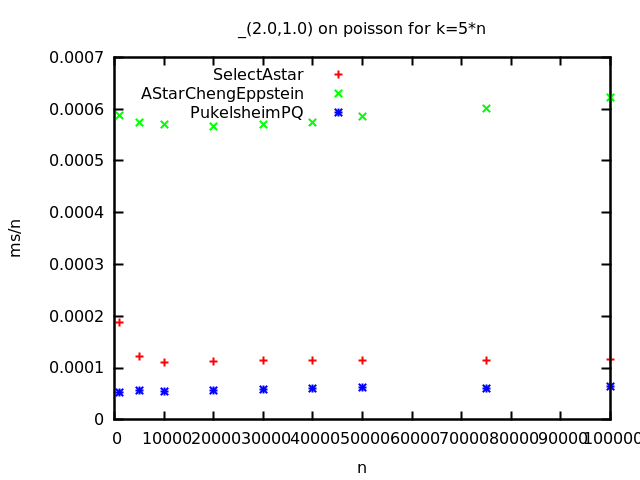}
\hfill\includegraphics[width=0.45\linewidth]{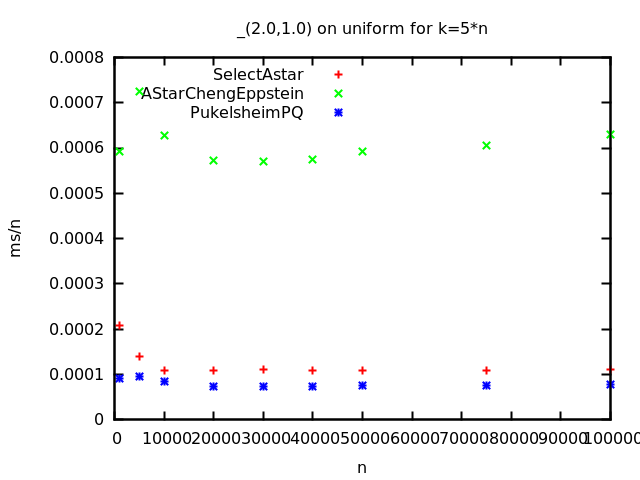}
\hfill

\clearpage\enlargethispage{\baselineskip}
Normalized runtimes of \CEalg for several input distributions and across several orders of magnitudes of $n$.

\includegraphics[width=0.45\linewidth]{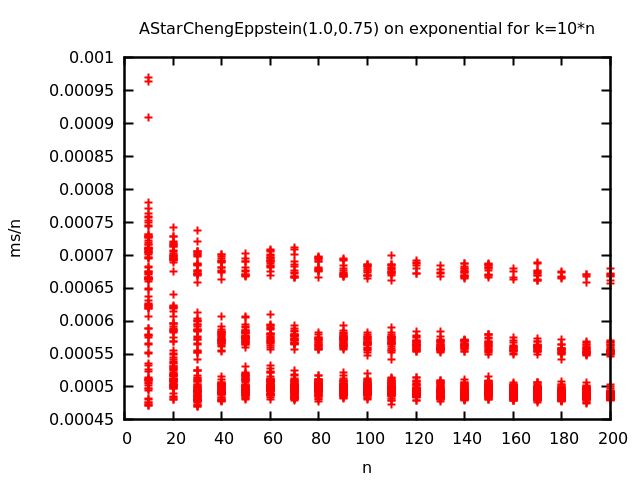}
\hfill\includegraphics[width=0.45\linewidth]{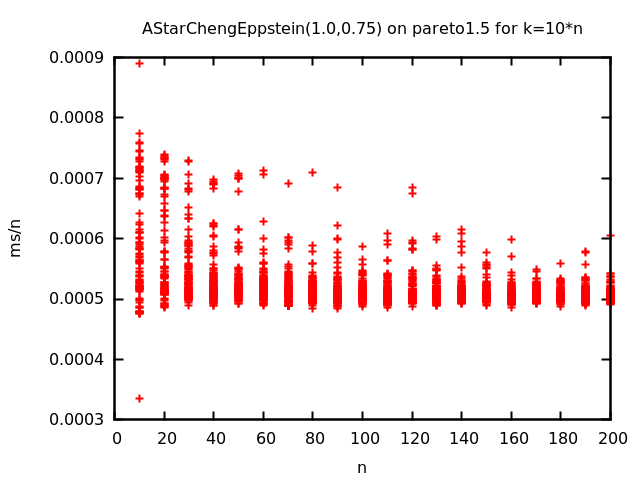}
\hfill\\
\hfill\includegraphics[width=0.45\linewidth]{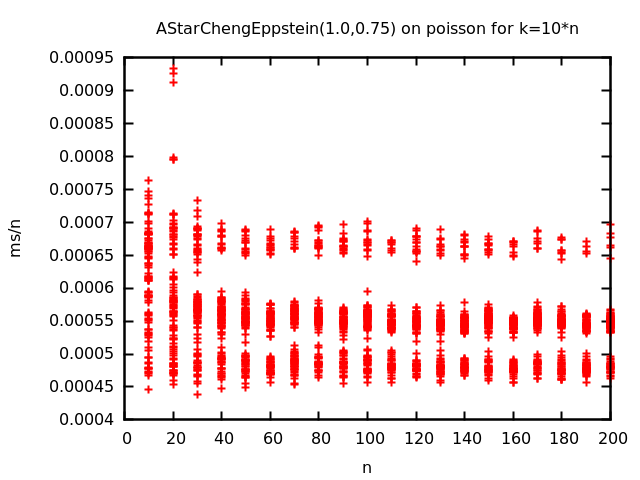}
\hfill\includegraphics[width=0.45\linewidth]{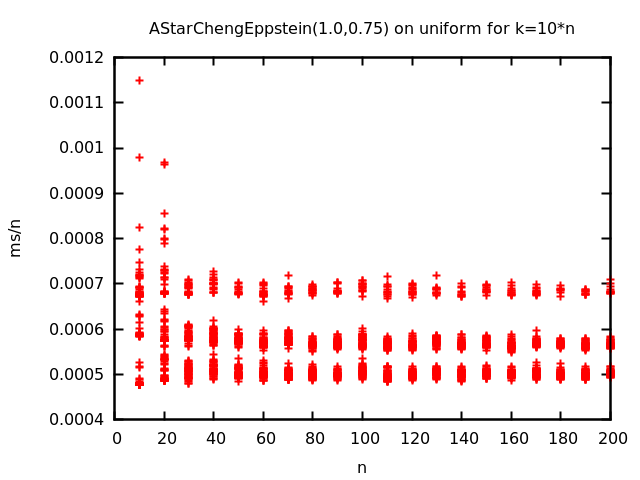}
\hfill\\
\hfill\includegraphics[width=0.45\linewidth]{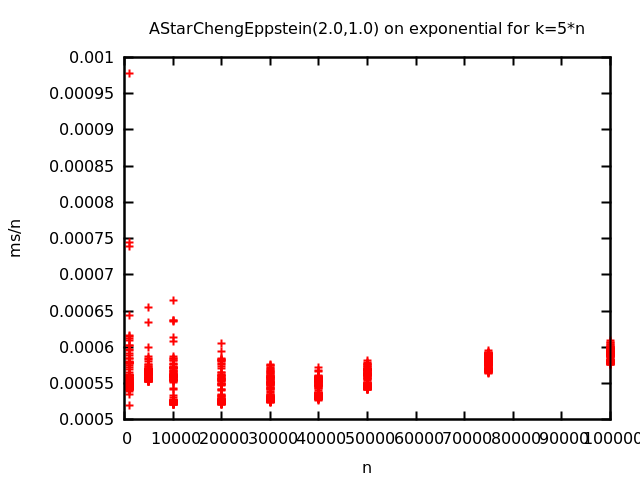}
\hfill\includegraphics[width=0.45\linewidth]{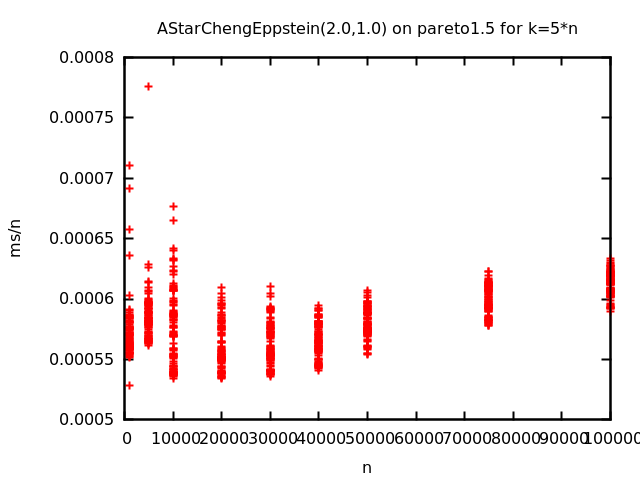}
\hfill\\
\hfill\includegraphics[width=0.45\linewidth]{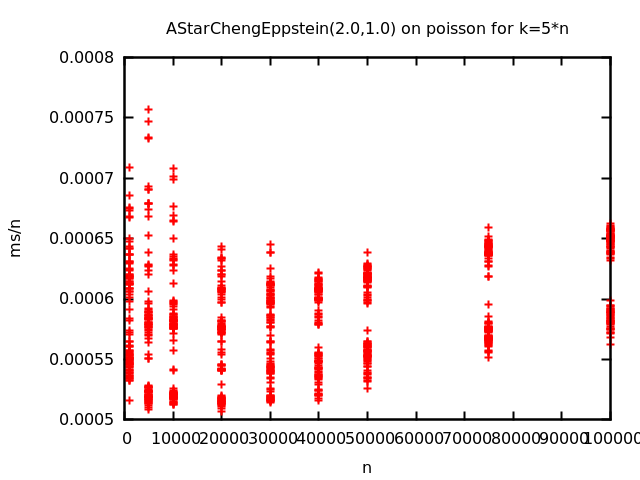}
\hfill\includegraphics[width=0.45\linewidth]{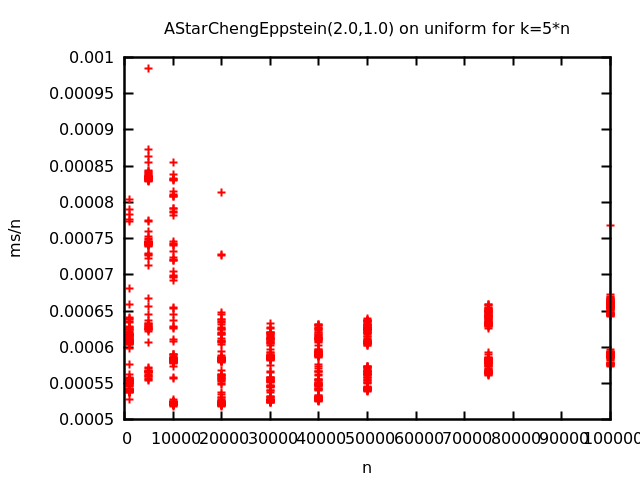}
\hfill

\clearpage\enlargethispage{\baselineskip}
Normalized runtimes of \Pukalg for several input distributions and across several orders of magnitudes of $n$.

\includegraphics[width=0.45\linewidth]{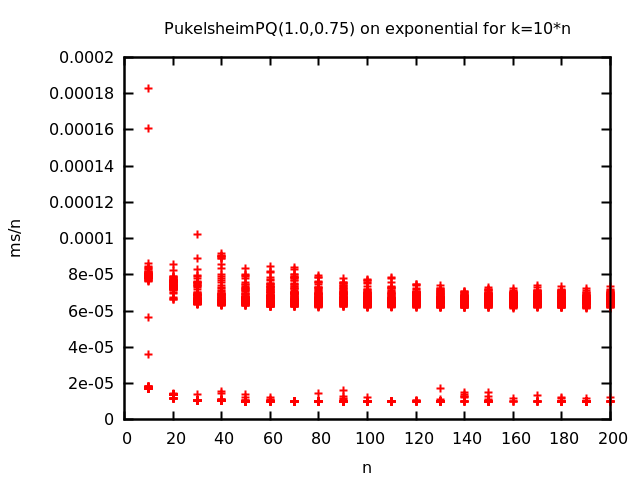}
\hfill\includegraphics[width=0.45\linewidth]{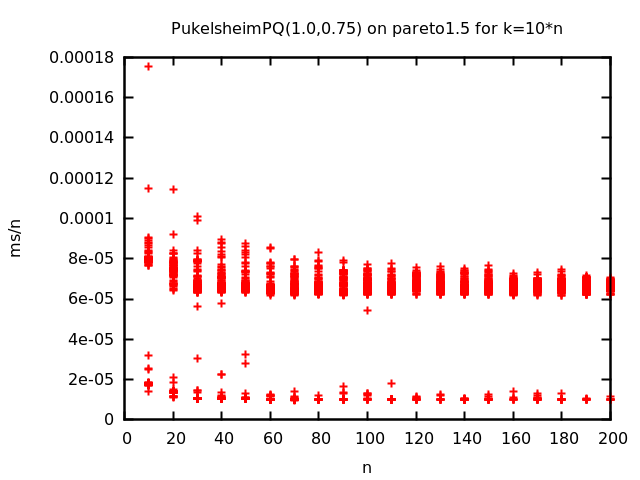}
\hfill\\
\hfill\includegraphics[width=0.45\linewidth]{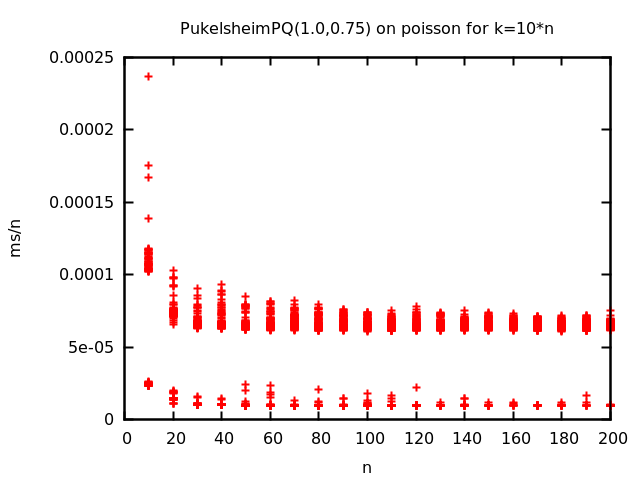}
\hfill\includegraphics[width=0.45\linewidth]{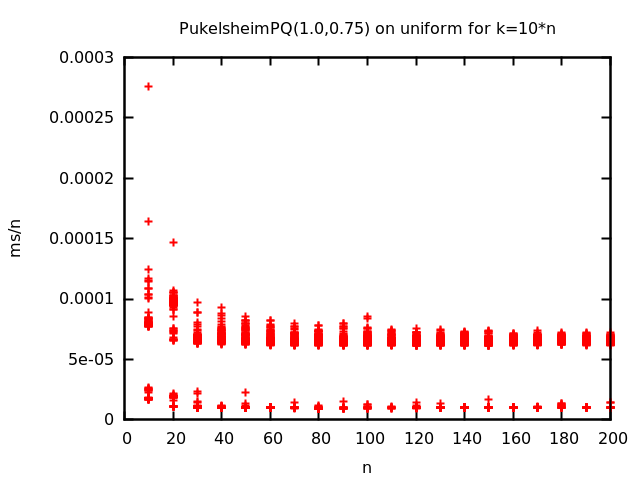}
\hfill\\
\hfill\includegraphics[width=0.45\linewidth]{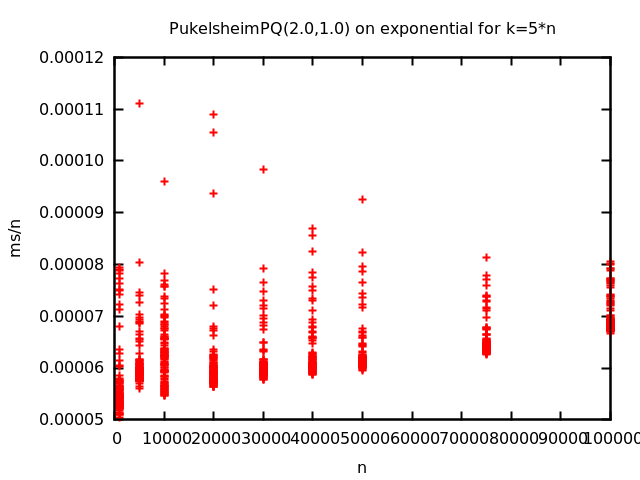}
\hfill\includegraphics[width=0.45\linewidth]{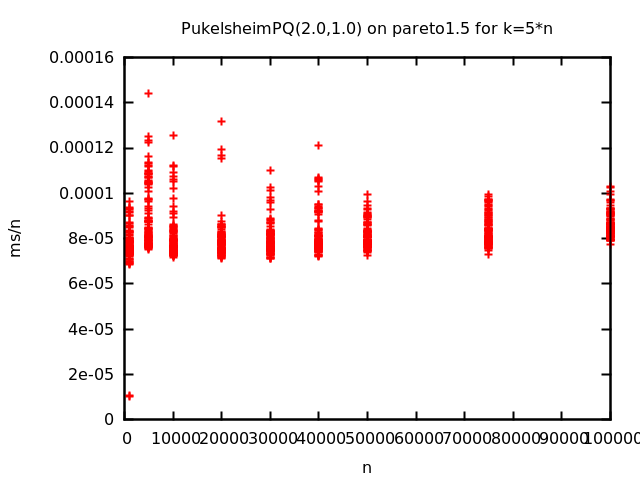}
\hfill\\
\hfill\includegraphics[width=0.45\linewidth]{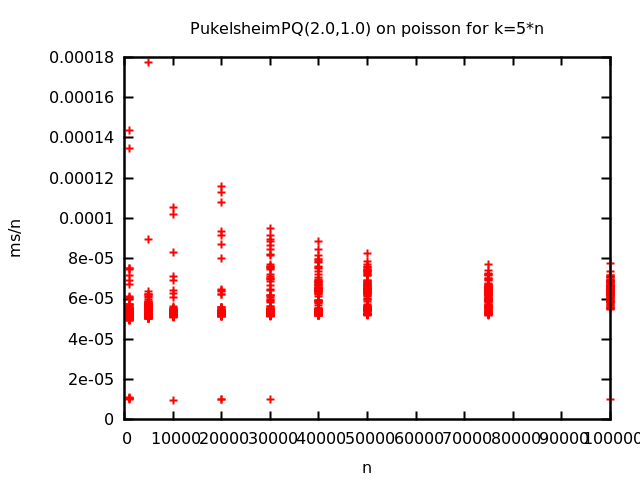}
\hfill\includegraphics[width=0.45\linewidth]{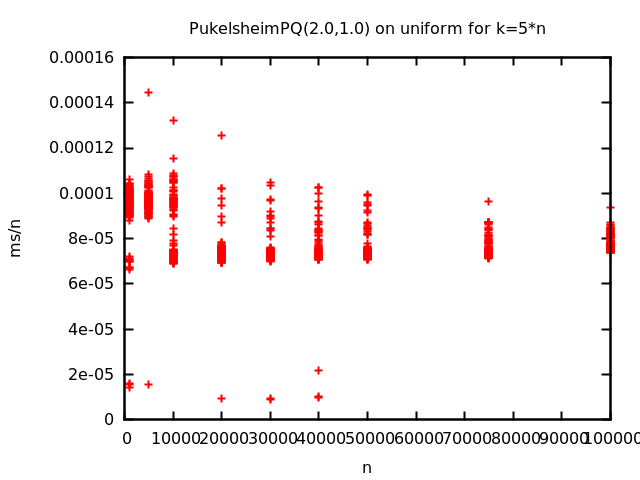}
\hfill

\clearpage\enlargethispage{\baselineskip}
Normalized runtimes of \RWalg for several input distributions and across several orders of magnitudes of $n$.

\includegraphics[width=0.45\linewidth]{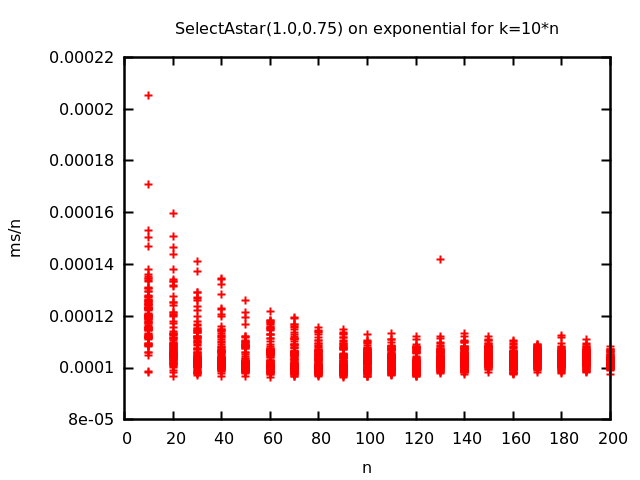}
\hfill\includegraphics[width=0.45\linewidth]{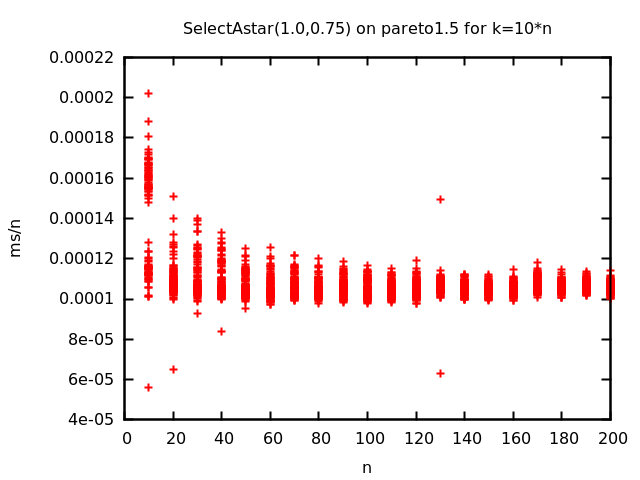}
\hfill\\
\hfill\includegraphics[width=0.45\linewidth]{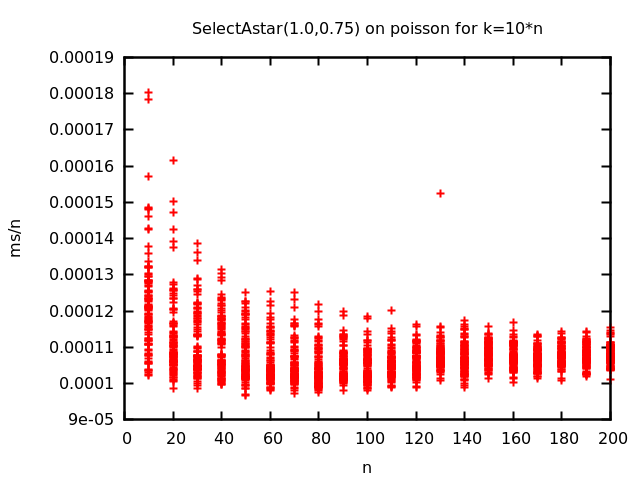}
\hfill\includegraphics[width=0.45\linewidth]{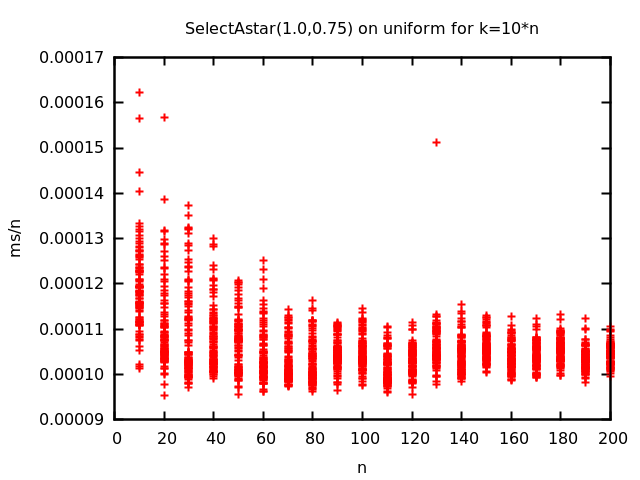}
\hfill\\
\hfill\includegraphics[width=0.45\linewidth]{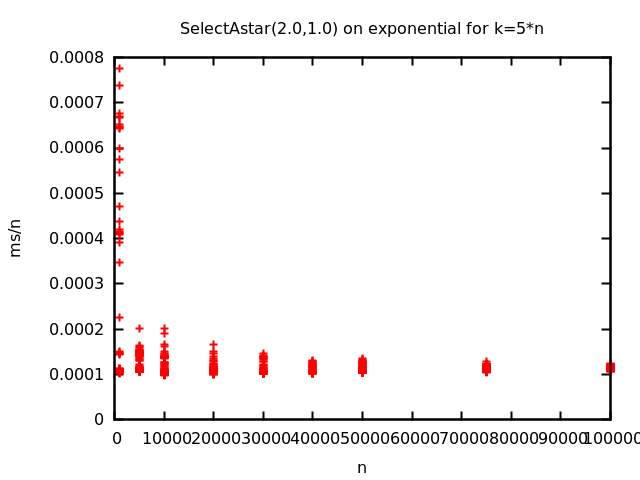}
\hfill\includegraphics[width=0.45\linewidth]{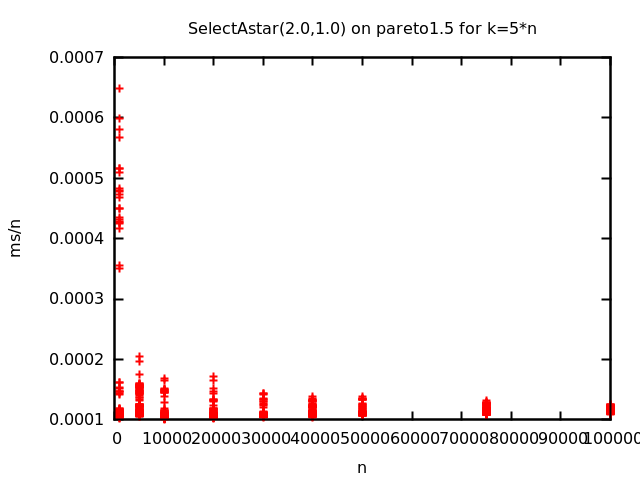}
\hfill\\
\hfill\includegraphics[width=0.45\linewidth]{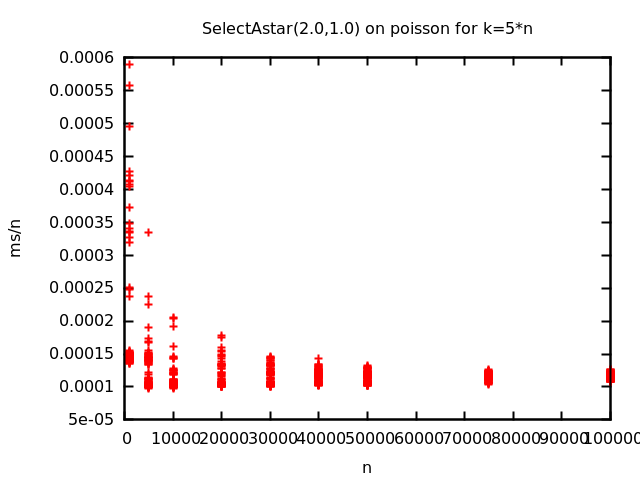}
\hfill\includegraphics[width=0.45\linewidth]{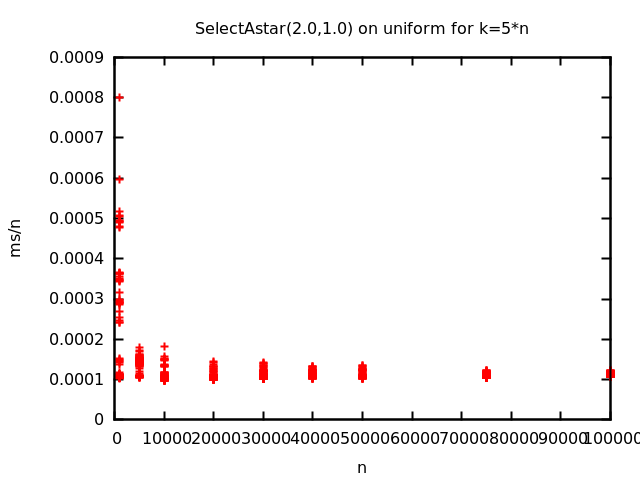}
\hfill

\clearpage\enlargethispage{\baselineskip}
Normalized $\Delta_a$ of \Pukalg for several input distributions and across several orders of magnitudes of $n$.

\includegraphics[width=0.45\linewidth]{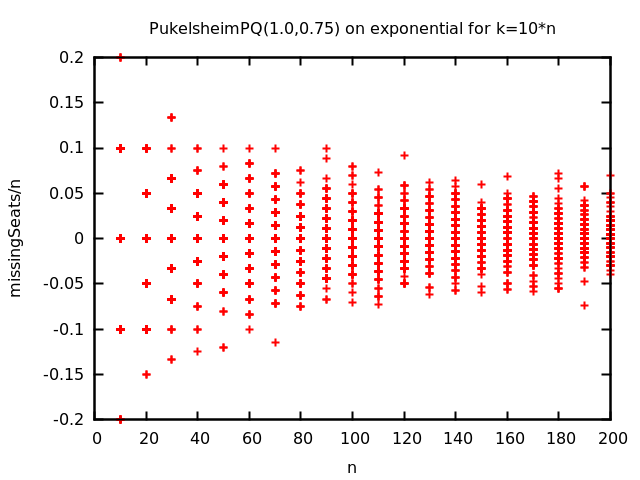}
\hfill\includegraphics[width=0.45\linewidth]{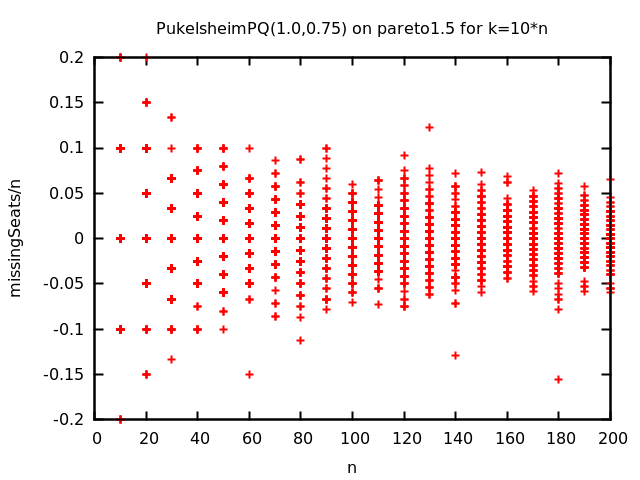}
\hfill\\
\hfill\includegraphics[width=0.45\linewidth]{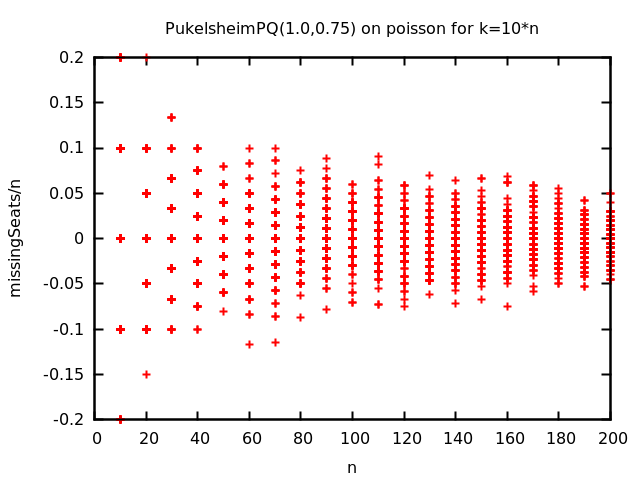}
\hfill\includegraphics[width=0.45\linewidth]{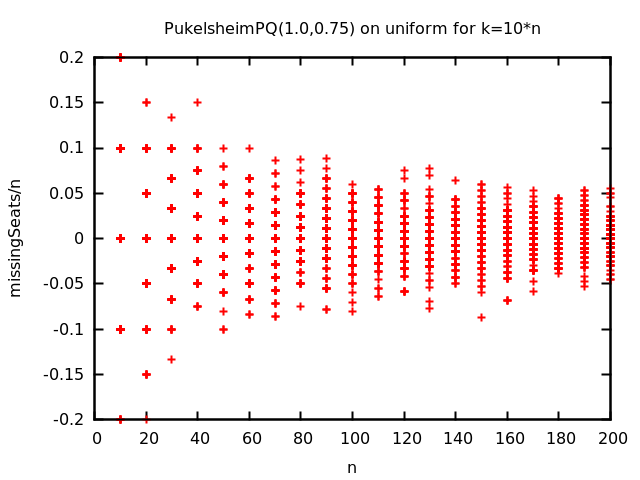}
\hfill\\
\hfill\includegraphics[width=0.45\linewidth]{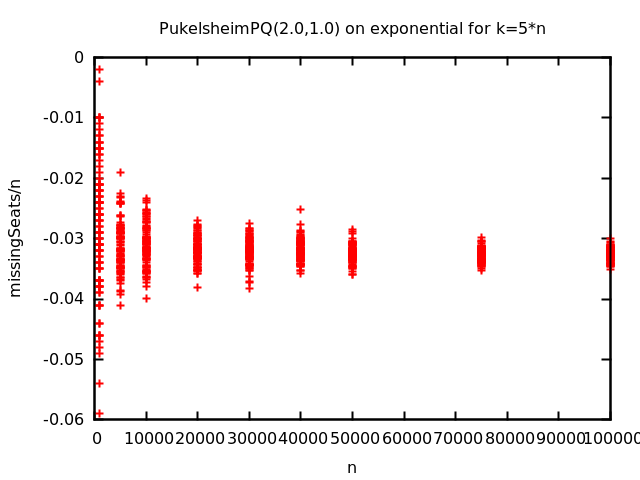}
\hfill\includegraphics[width=0.45\linewidth]{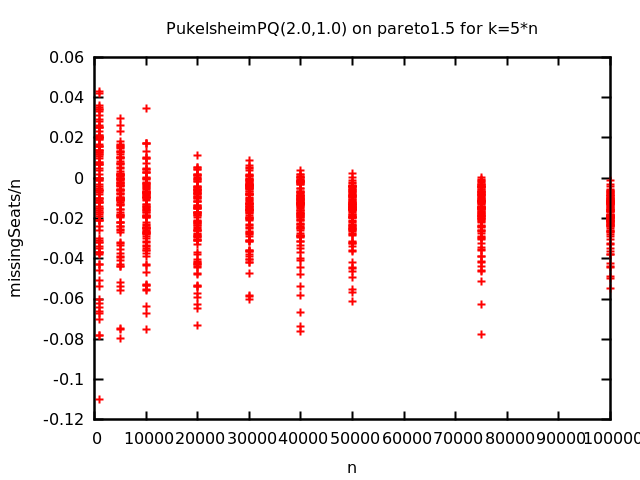}
\hfill\\
\hfill\includegraphics[width=0.45\linewidth]{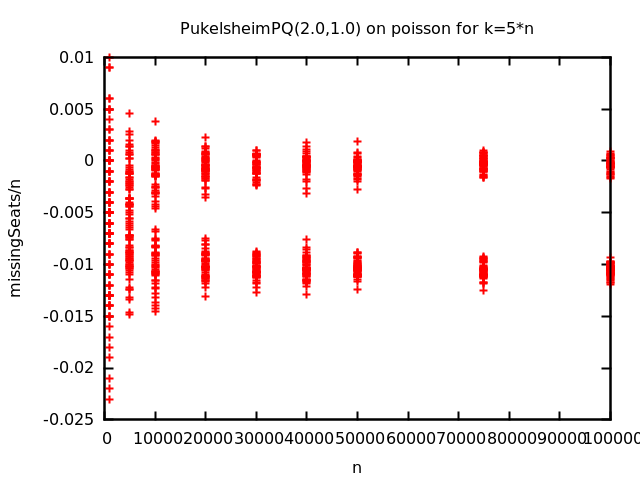}
\hfill\includegraphics[width=0.45\linewidth]{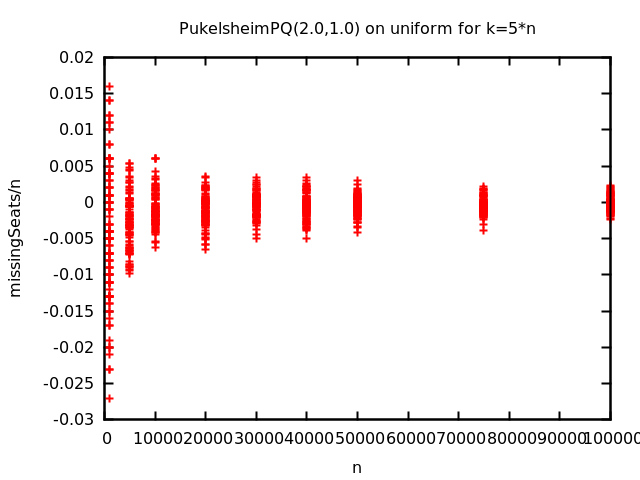}
\hfill

\clearpage\enlargethispage{\baselineskip}
Normalized $|\hat{\mset{A}}|$ of \RWalg for several input distributions and across several orders of magnitudes of $n$.

\includegraphics[width=0.45\linewidth]{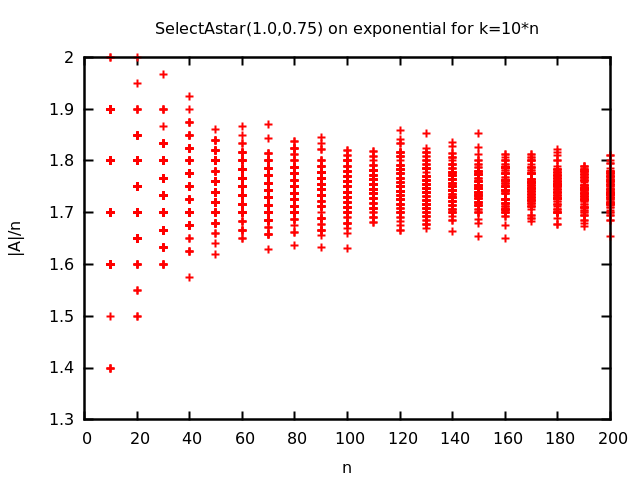}
\hfill\includegraphics[width=0.45\linewidth]{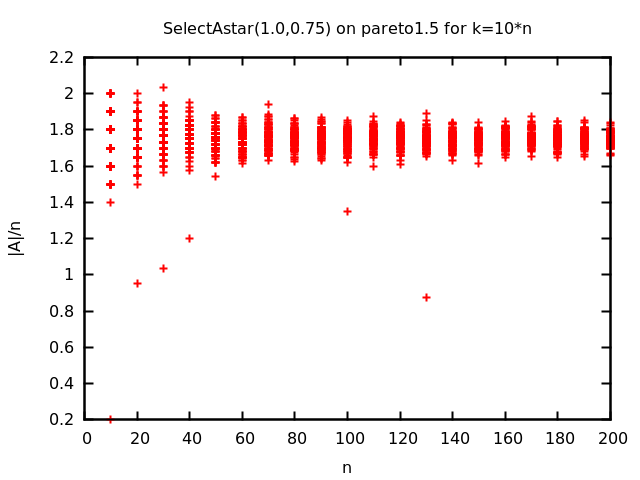}
\hfill\\
\hfill\includegraphics[width=0.45\linewidth]{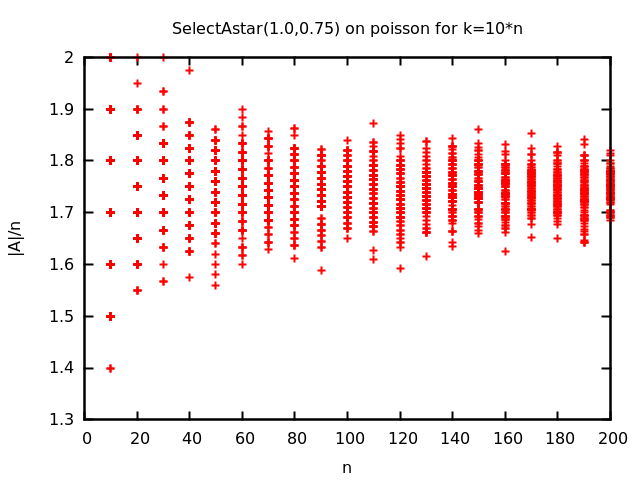}
\hfill\includegraphics[width=0.45\linewidth]{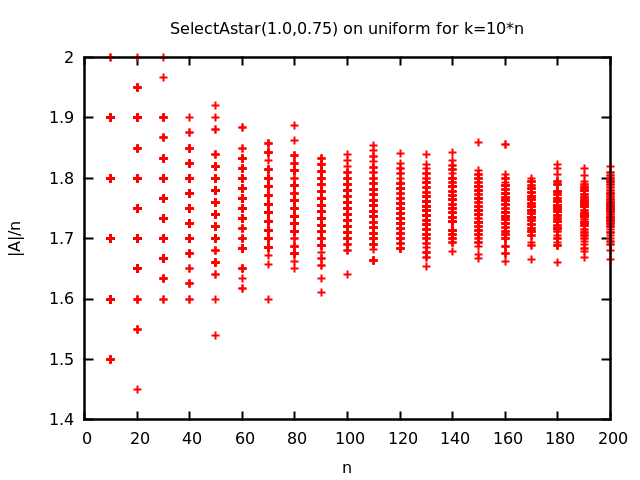}
\hfill\\
\hfill\includegraphics[width=0.45\linewidth]{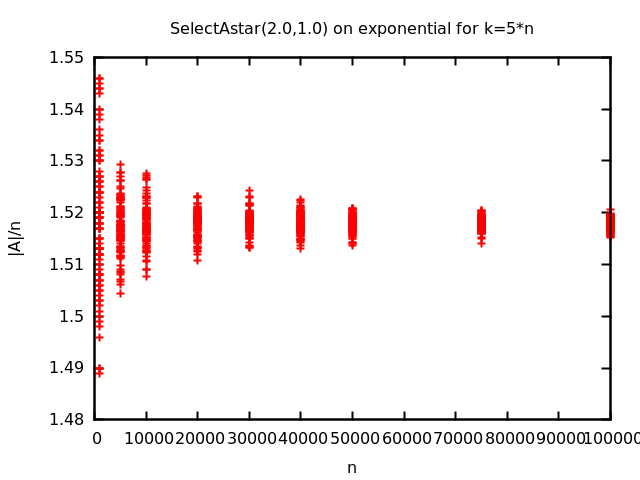}
\hfill\includegraphics[width=0.45\linewidth]{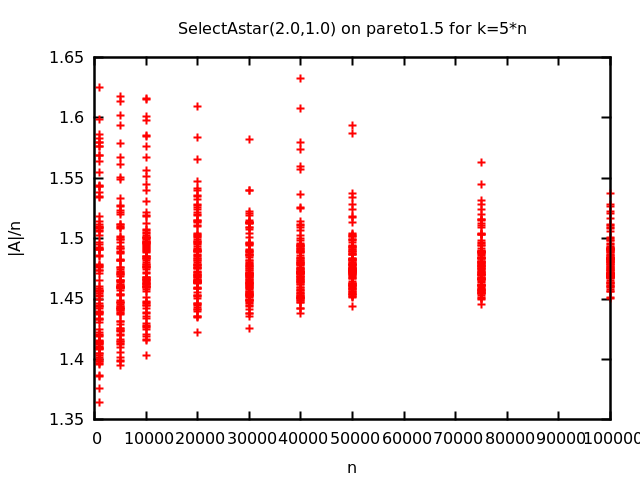}
\hfill\\
\hfill\includegraphics[width=0.45\linewidth]{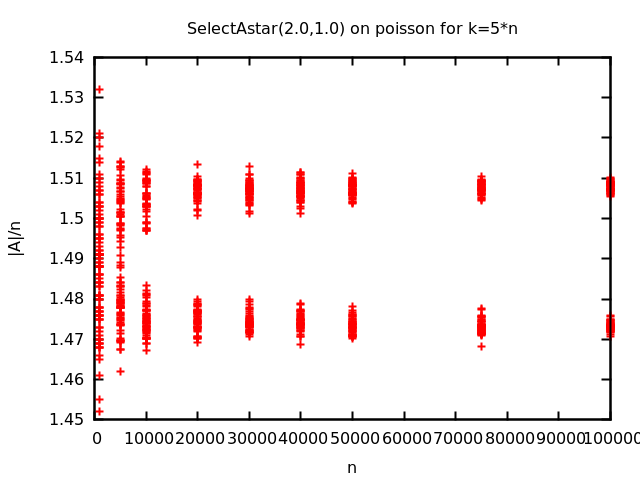}
\hfill\includegraphics[width=0.45\linewidth]{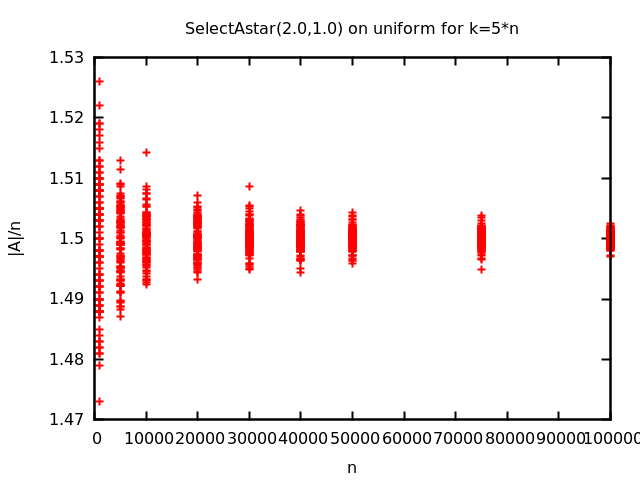}
\hfill

\clearpage\enlargethispage{\baselineskip}
Runtimes against $\Delta_a$ of \Pukalg for several input distributions and across several orders of magnitudes of $n$. Each color stands for one $n$.

\includegraphics[width=0.45\linewidth]{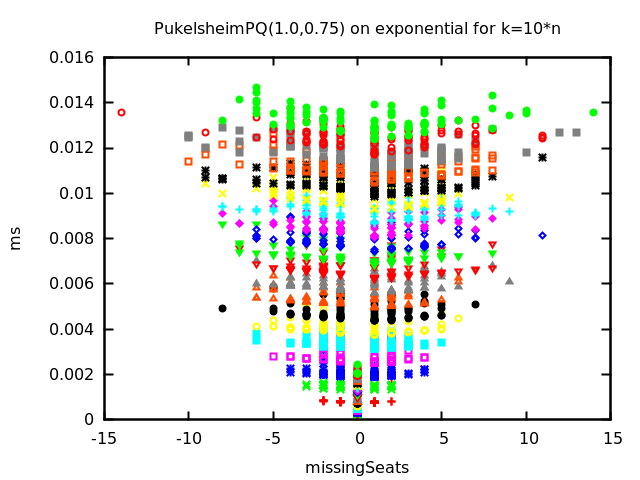}
\hfill\includegraphics[width=0.45\linewidth]{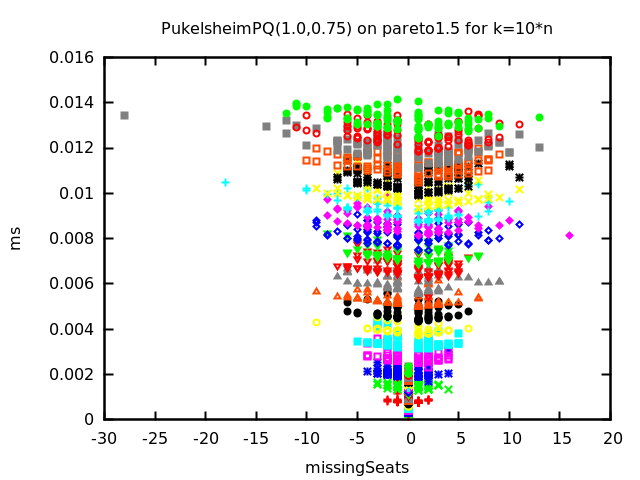}
\hfill\\
\hfill\includegraphics[width=0.45\linewidth]{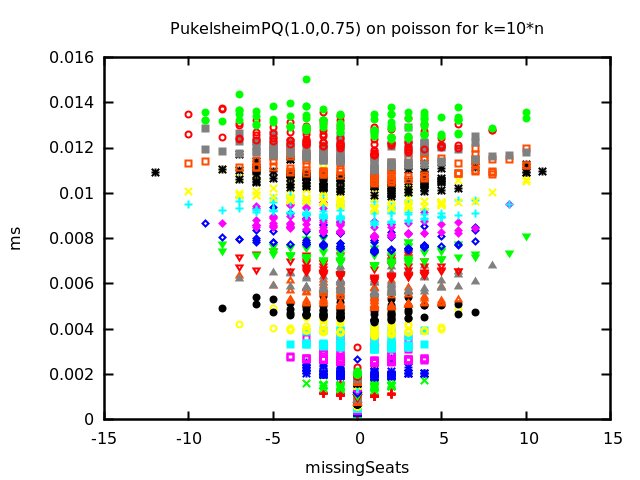}
\hfill\includegraphics[width=0.45\linewidth]{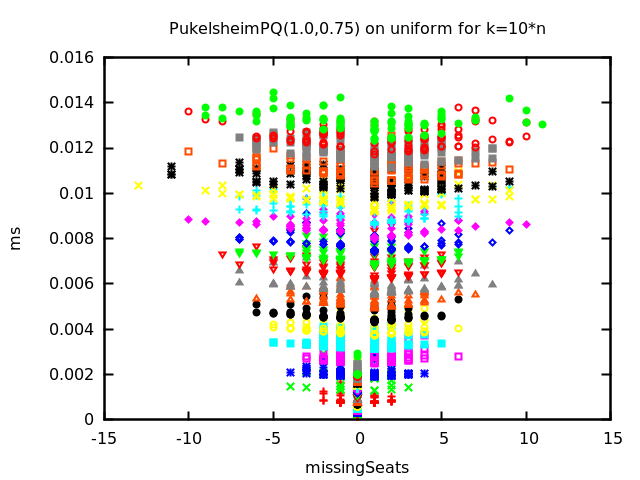}
\hfill\\
\includegraphics[width=0.45\linewidth]{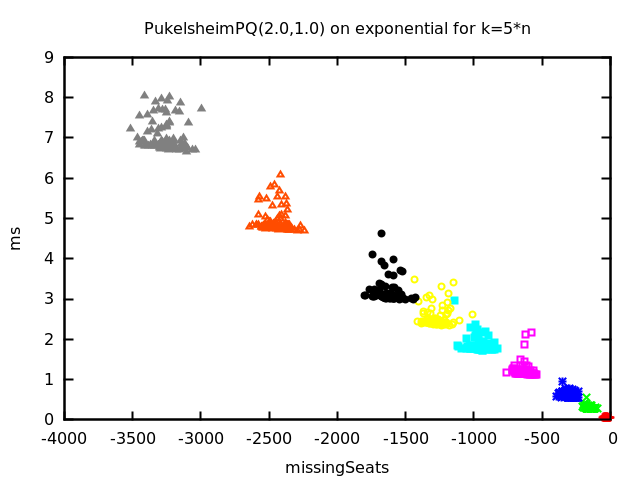}
\hfill\includegraphics[width=0.45\linewidth]{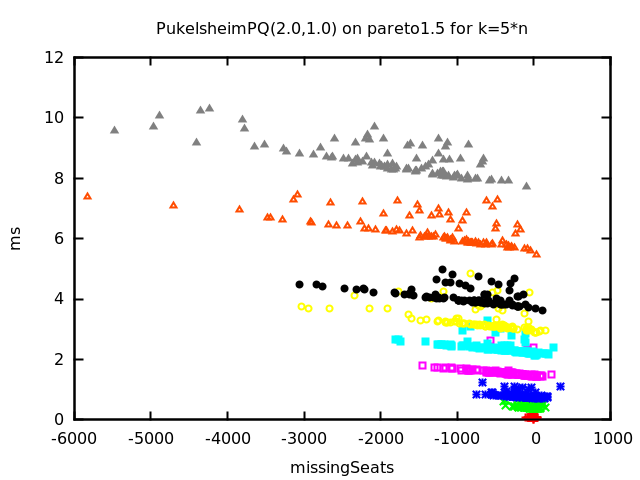}
\hfill\\
\hfill\includegraphics[width=0.45\linewidth]{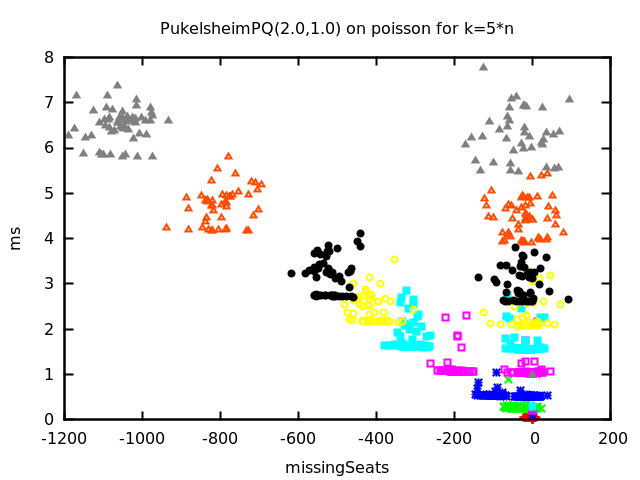}
\hfill\includegraphics[width=0.45\linewidth]{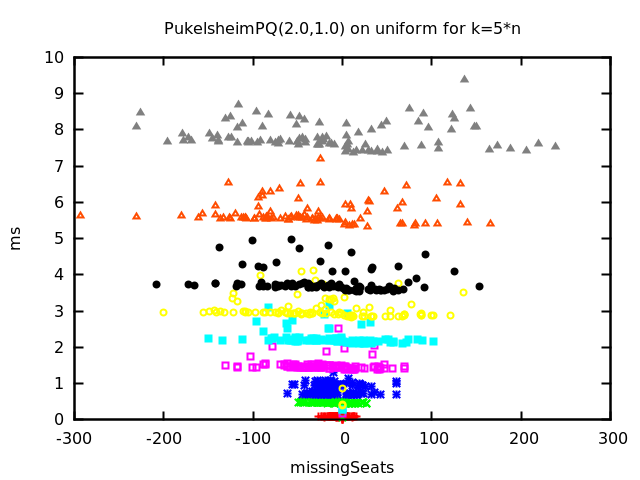}
\hfill

\clearpage\enlargethispage{\baselineskip}
Runtimes against $|\hat{\mset{A}}|$ of \RWalg for several input distributions and across several orders of magnitudes of $n$. Each color stands for one $n$.

\includegraphics[width=0.45\linewidth]{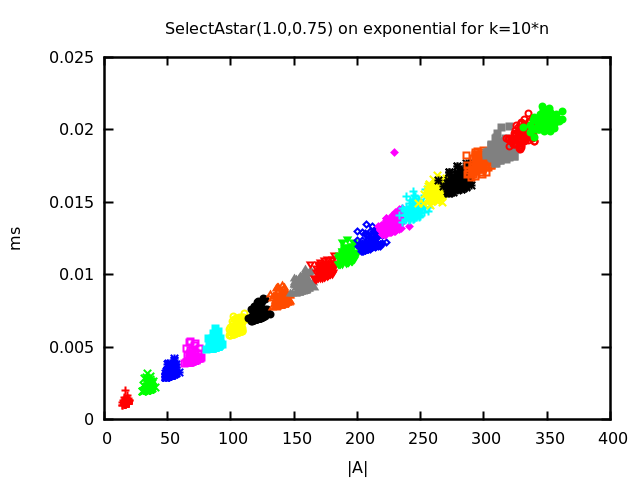}
\hfill\includegraphics[width=0.45\linewidth]{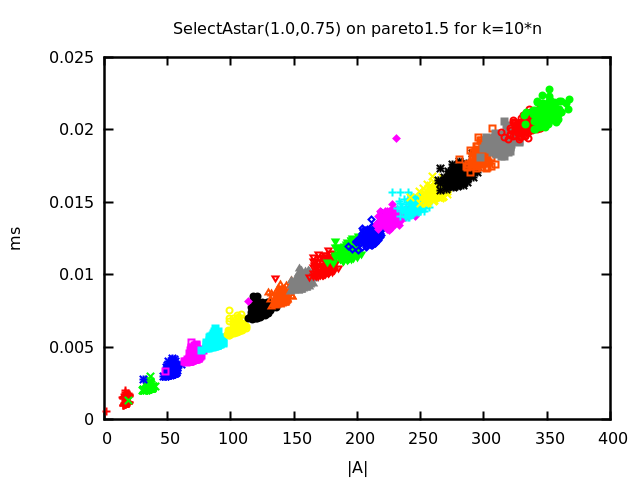}
\hfill\\
\hfill\includegraphics[width=0.45\linewidth]{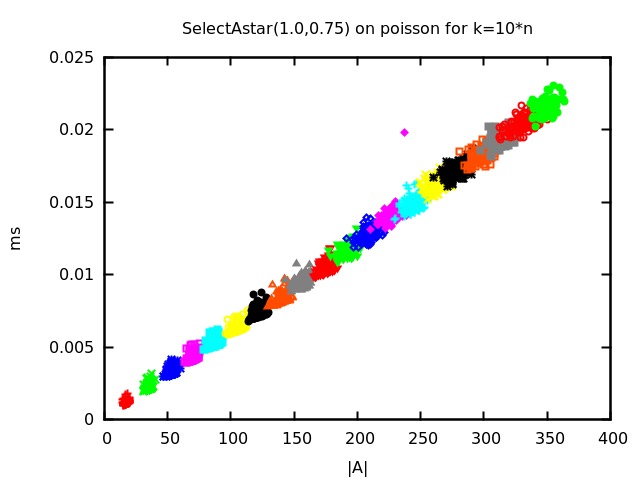}
\hfill\includegraphics[width=0.45\linewidth]{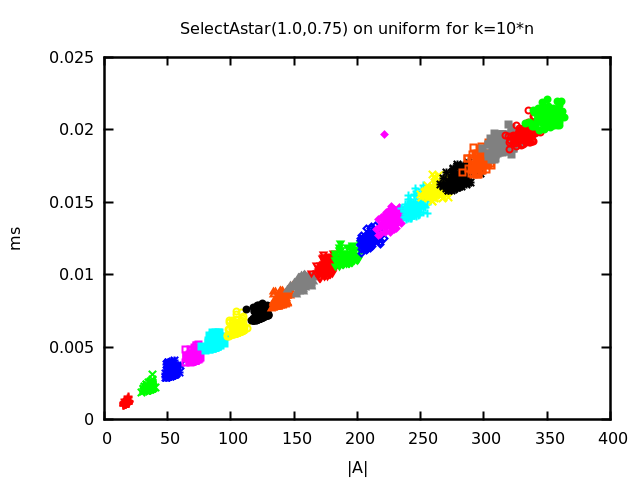}
\hfill\\
\hfill\includegraphics[width=0.45\linewidth]{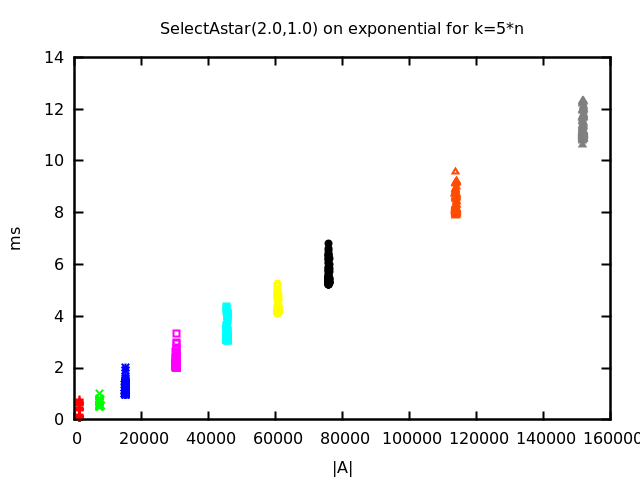}
\hfill\includegraphics[width=0.45\linewidth]{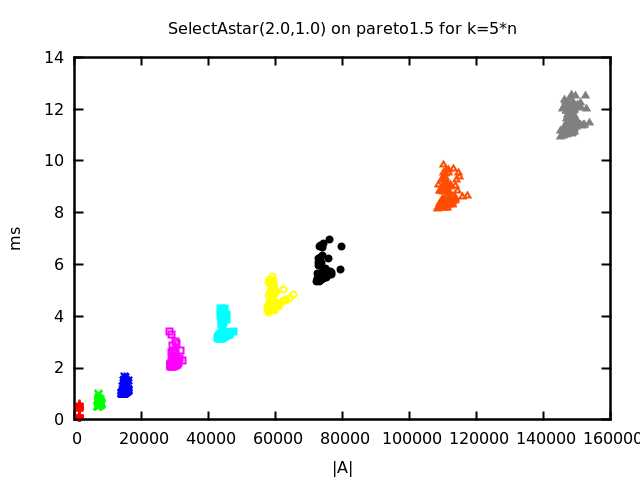}
\hfill\\
\hfill\includegraphics[width=0.45\linewidth]{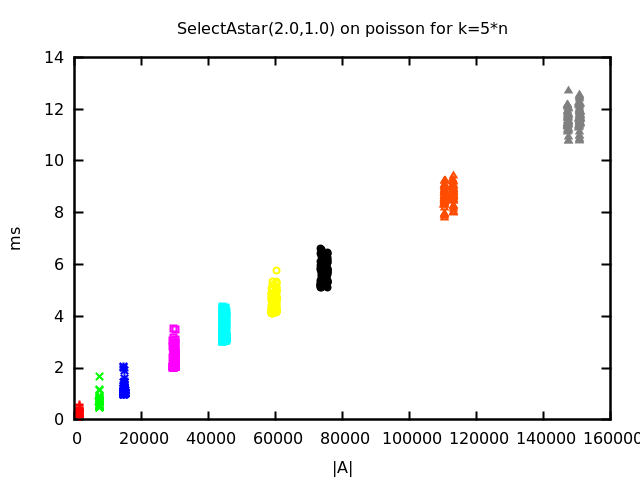}
\hfill\includegraphics[width=0.45\linewidth]{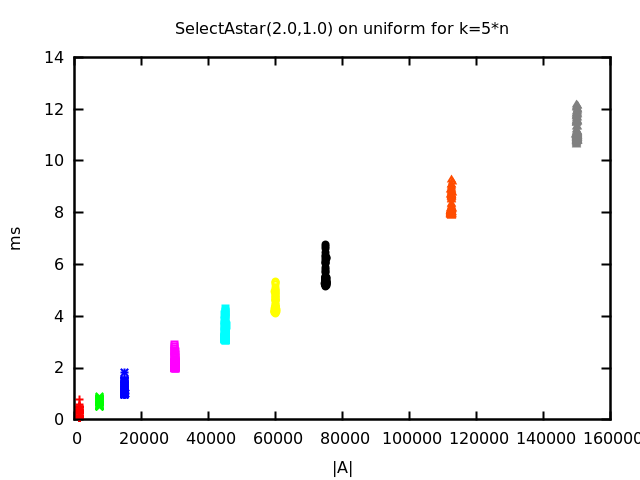}

%% file: app-notation-reference.tex
\section{Index of Used Notation}
\label{app:notations}

In this section, we collect the notation used in this paper.
Some might be seen as ``standard'', but we think
including them here hurts less than a potential 
misunderstanding caused by omitting them.

\subsection*{Generic Mathematical Notation}
\begin{notations}
\notation{$\lfloor x\rfloor$, $\lceil x\rceil$}
	floor and ceiling functions, as used in \cite{ConcreteMathematics}.
\notation{$M_{(k)}$}
	The $k$th~smallest element of (multi)set/vector $M$ (assuming it exists);\\
	if the elements of $M$ can be written in non"=decreasing order, $M$ is  
	given by $M_{(1)} \le M_{(2)} \le M_{(3)} \le \cdots$.\\
	\textsl{Example:} For $M = \{ 5,8,8,8,10,10 \}$, we have $M_{(1)} = 5$, 
	$M_{(2)} = M_{(3)} = M_{(4)} = 8$, and
	$M_{(5)} = M_{(6)} = 10$.
\notation{$M^{(k)}$}
	Similar to $M_{(k)}$, but $M^{(k)}$ denotes the $k$th \emph{largest} element.
\notation{$\vect x = (x_1,\ldots,x_d)$}
	to emphasize that $\vect x$ is a vector, it is written in \textbf{bold};\\
	components of the vector are written in regular type.
\notation{$\mset M$}
	to emphasize that $\mset M$ is a multiset, it is written in calligraphic 
	type.
\notation{$\mset{M}_1 \uplus \mset{M}_2$}
	multiset union; multiplicities add up.
\end{notations}

\subsection*{Notation Specific to the Problem}
\begin{notations}
\notation{party, seat, vote (count), chamber size}
	Parties are assigned seats (in parliament), so that 
	the number of seats $s_i$ that party $i$ is assigned is 
	(roughly) proportional to that party's vote count $v_i$
	and the overall number of assigned seats equals the chamber size $k$.
\notation{$d = (d_j)_{j=0}^\infty$}
	the divisor sequence used in the highest averages method;
	$d$ must be a nonnegative, (strictly) increasing and unbounded sequence.
\notation{$\delta$, $\delta^{-1}$}
  a continuation of $j \mapsto d_j$ on the reals and its inverse, 
  both of which can be evaluated in constant time.
\notation{$n$}
	number of parties in the input.
\notation{$\vect v$, $v_i$}
	$\vect v = (v_1,\ldots,v_n) \in \Q_{>0}^n$, 
	vote counts of the parties in the input.
\notation{$V$}
  the sum $v_1 + \cdots + v_n$ of all vote counts.
\notation{$k$}
	$k \in \N$, the number of seats to be assigned; also called house size.
\notation{$\vect s$, $s_i$}
	$\vect s = (s_1,\ldots,s_n) \in \N_0$, 
	the number of seats assigned to the respective parties; the result.
\notation{$a_{i,j}$}
	$a_{i,j} \ce d_j / v_i$, the ratio used to define divisor methods; 
	$i$ is the party, $j$ is the number of seats 
	$i$ has already been assigned.\\
\notation{$A_i$}	
	For party $i$, $A_i \ce \{ a_{i,0}, a_{i,1}, a_{i,2}, \ldots \}$
	is the list of (reciprocals of) party~$i$'s ratios.\\
\notation{$a$}
	We use $a$ as a free variable when an arbitrary $a_{i,j}$ is meant.
\notation{$\mathcal A$}
	$\mathcal A \ce A_1 \uplus \cdots \uplus A_n$
	is the multiset of all averages.
\notation{$r(x, \mathcal{A})$}
  the rank of $x$ in $\mathcal{A}$, that is the number of elements in
  multiset $\mset A$ that are no larger than $x$; $r(x)$ for short if
  $\mset A$ is clear from context.
\notation{$a^*$}
	the ratio $a^* = a_{i^*\!\!,j^*}$ selected for assigning the last 
	(i.\,e.\ the $k$th) seat; corresponds to $\vect s$ by $s_i = r(a^*, A_i)$;
	$a^* = \mathcal A_{(k)}$ (cf.\ \wref{sec:notation} and \wref{sec:fast-selection-alg}).
\notation{$\overline x$}
	an upper bound $\overline x > a^*$; we use $\overline x = d_{k-1}/v_1 + \varepsilon$,
	where $\varepsilon > 0$ is a suitable constant.
\notation{$I_{\overline x}$}
	$I_{\overline x} \ce \{ i \mid v_i > d_0 / \overline x\} $; 
	the set of parties $i$ whose vote count is large enough, so that $a_{i,0} < \overline x$,
	i.\,e.\ so that they contribute to the rank of $\overline x$ in $\mathcal A$.
\notation{$V_{\overline x}$}
  the sum of the vote counts of all parties in $I_{\overline{x}}$.
\notation{$\mset{A}^{\overline{x}}$}
  the elements in $\mset A$ that are smaller than $\overline{x}$, 
  i.\,e., $\mset A \cap (-\infty, \overline{x})$.
\notation{$\underline a$, $\overline a$}
	lower and upper bounds on candidates $\underline a \le a\le \overline a$
	such that still $a^* \in \mathcal A \cap [\underline a, \overline a]$.
\end{notations}

%% file: app-errata.tex
\section{Changelog}
\label{app:errata}

The following (substantial) changes have been made from arXiv version~2 to~3.
\begin{itemize}
  \item \wref{lem:candidate-set-linear-corridor} has been strengthened;
    both $\underline a$ and the upper bound on $|\mset A \cap [\underline a, \overline a]|$
    have been improved. Both changes are due to the observation that we could
    require $\underline\beta \leq \alpha$ without loss of generality.

    Related notation update: $(\check\beta, \beta) \leadsto (\underline\beta, \overline\beta)$.
    
  \item We have added \wref{app:methods-scope} in order to clarify that the
    assumptions we make for our main result do restrict the scope of divisor
    methods we cover by too much.
\end{itemize}